\documentclass[twocolumn, twoside]{IEEEtran}
\usepackage[table]{xcolor}
\usepackage{amssymb,epsfig,graphics, graphicx, url,times,mathrsfs,algorithmic,amsmath,array,latexsym,fancyhdr,xspace,wrapfig, xcolor,multirow,amsthm,caption,cite,helvet,sidecap,bm,hyperref,url,dsfont,nicefrac, subcaption,bbm, xargs, adjustbox}
\usepackage[capitalise]{cleveref}

\usepackage[normalem]{ulem}

\usepackage[inline]{enumitem}
\usepackage{arydshln}
\usepackage[makeroom]{cancel}
\usepackage{varwidth}
\usepackage{pgfplots}
\usepackage[ruled,vlined]{algorithm2e}
\usepackage{balance}

\usepackage{tikz}
\usetikzlibrary{shapes.geometric}
\usetikzlibrary{shapes,snakes,patterns}
\usetikzlibrary{arrows,positioning,automata,calc}
\usetikzlibrary{plotmarks}

\definecolor{green1}{rgb}{0,0.5,0}
\definecolor{ogreen}{rgb}{0,0.5,0}
\definecolor{magenta}{rgb}{1.0, 0.11, 0.81}
\definecolor{mulberry}{rgb}{0.77, 0.29, 0.55}
\definecolor{xgray}{rgb}{0.9, 0.9, 0.9}
\definecolor{eggplant}{rgb}{0.38, 0.25, 0.32}
\definecolor{electriccrimson}{rgb}{1.0, 0.0, 0.25}
\definecolor{bostonuniversityred}{rgb}{0.8, 0.0, 0.0}

\def \blue{\color{blue}}

\pgfplotsset{compat=1.16}

\newcommand{\temp}{\textcolor{red}}
\newcommand{\master}{main node\xspace}
\newcommand{\worker}{worker\xspace}
\newcommand{\Worker}{Worker\xspace}
\newcommand{\Workers}{Workers\xspace}
\newcommand{\workers}{workers\xspace}
\usetikzlibrary{decorations.pathreplacing,calc}
\tikzset{brace/.style={decorate, decoration={brace}},
 brace mirrored/.style={decorate, decoration={brace,mirror}},
}

\newcounter{brace}
\newcommand{\PreserveBackslash}[1]{\let\temp=\\#1\let\\=\temp}
\newcolumntype{C}[1]{>{\PreserveBackslash\centering}p{#1}}

\definecolor{green1}{rgb}{0,0.5,0}
\definecolor{magenta}{rgb}{1.0, 0.11, 0.81}
\definecolor{mulberry}{rgb}{0.77, 0.29, 0.55}
\definecolor{xgray}{rgb}{0.9, 0.9, 0.9}
\def \blue{\color{blue}}

\newcommand{\stragglerun}{\ensuremath \varrho_1}
\newcommand{\stragglertr}{\ensuremath \varrho_2}
\newcommand{\noworkersuntrusted}{\ensuremath n_1}
\newcommand{\noworkerstrusted}{\ensuremath n_2}

\def \bes{\begin{equation*}}
\def \ees{\end{equation*}}
\def \bas{\begin{align*}}
\def \eas{\end{align*}}
\def \be{\begin{equation}}
\def \ee{\end{equation}}
\def \bbm{\begin{bmatrix}}
\def \ebm{\end{bmatrix}}

\def \rvA{\texttt{A}}
\def \rvB{\texttt{B}}
\def \rvC{\texttt{C}}

\def \rvR {\texttt{R}}

\def \rvT{\texttt{T}}
\def \rvX{\texttt{X}}
\def \rvY{\texttt{Y}}

\def \F{\mathbb{F}}

\newcommand{\Aij}{\rvA_{\{i,j\}}}

\newcommand{\Rij}{\rvR_{\{i,j\}}}

\newcommand{\Xij}[2]{\rvX_{\{#1,#2\}}}

\newcommand{\thresh}{\ensuremath{t}}

\newcommand{\sss}{\ensuremath{\mathrm{SecShare}}}
\newcommand{\sparsess}{\ensuremath{\mathrm{SparseSecShare}}}

\newcommand{\cA}{{\cal A}}

\newcommand{\cC}{{\cal C}}

\newcommand{\cP}{{\cal P}}

\newcommand{\cW}{{\cal W}}

\newcommand{\cZ}{{\cal Z}}

\newcommand{\bfx}{{\boldsymbol x}}

\newcommand{\bfA}{{\mathbf A}}
\newcommand{\bfB}{{\mathbf B}}
\newcommand{\bfC}{{\mathbf C}}

\newcommand{\bfP}{{\mathbf P}}

\newcommand{\bfR}{{\mathbf R}}
\newcommand{\bfS}{{\mathbf S}}
\newcommand{\bfT}{{\mathbf T}}

\newcommand{\bfX}{{\mathbf X}}

\newcommand{\rvTask}{\widetilde{\rvA}}
\newcommand{\task}{\widetilde{\bfA}}

\newcommand{\bfXij}{\bfX_{\{i,j\}}}
\newcommand{\I}{\textrm{I}}
\newcommand{\mutinf}{\I_q}
\newcommand{\kl}[2]{\text{D}_{\text{KL}} \left(#1 \Vert #2 \right)}
\newcommand{\entropy}{\textrm{H}_q}
\newcommand{\leakage}[1]{\textrm{L}_#1}

\newcommand{\sleak}{\ensuremath{\varepsilon}}
\newcommand{\relleakage}{\ensuremath{\bar{\sleak}}}
\newcommand{\spars}{\ensuremath{s}}
\newcommand{\stragglernum}{\ensuremath{\sigma}}
\newcommand{\sr}{\ensuremath s_{\bfR}}
\newcommand{\sa}{\ensuremath s_{\bfA}}
\newcommand{\sai}{\ensuremath s_{\bfA_i}}

\newcommand{\sasi}{\ensuremath s_{\task_i}}
\newcommand{\sbm}{\ensuremath s_{\bfB}}
\newcommand{\sar}{\ensuremath s_{\bfA+\bfR}}
\newcommand{\savg}{\ensuremath s_{\text{avg}}}
\newcommand{\sdiff}{\ensuremath s_{\delta}}

\newcommand{\sparsity}{\textrm{S}}
\newcommand{\pz}{p_{1}}
\newcommand{\pzinv}{p_{1}^{\text{inv}}}
\newcommand{\pnz}{p_{3}}
\newcommand{\pext}{p_{2}}
\newcommand{\pnzinv}{p_{2,3}^\text{inv}}
\newcommand{\ptwoinv}{p_{2,3}^\text{inv}}

\newtheorem{theorem}{Theorem}

\newtheorem{lemma}{Lemma}

\newtheorem{proposition}{Proposition}
\newtheorem{construction}{Construction}
\newtheorem{definition}{Definition}

\newtheorem{remark}{Remark}

\usetikzlibrary{decorations.pathreplacing,calc}
\newcommand{\Fq}{\mathbb{F}_q}
\newcommand{\Fqstar}{\Fq^\star} 
\newcommand{\pa}[1][a]{\mathrm{P}_{\rvA}}
\newcommand{\pb}[1][b]{\mathrm{P}_{\rvB}}
\newcommand{\paandapr}{\mathrm{P}_{\rvA,\rvA+\rvR}}
\newcommand{\papr}{\mathrm{P}_{\rvA+\rvR}}

\newcommand{\paandr}[1][a]{\mathrm{P}_{\rvA,\rvR}}
\newcommand{\pr}[1][r]{\mathrm{P}_{\rvR}}
\newcommand{\optimizer}{\min\limits_{\mathcal{P}}} %
\newcommandx{\pra}[2][1=r, 2=a]{p_{#1\lvert#2}}
\newcommand{\lossopt}{\mathrm{L}_\text{opt}}
\newcommand{\totalleakage}{\mathrm{L}(\pz,\pzinv,\pext,\pnz,\pnzinv)}
\newcommand{\totalleakagesdiff}{\mathrm{L}(\pz,\pzinv,\pext,\pnz,\pnzinv,\sdiff)}

\newcommand{\loss}{\mathcal{L}(\pz,\pzinv,\pext,\pnz,\pnzinv,\lambda_1,\lambda_2,\lambda_3,\lambda_4)}
\newcommand{\lossabbrev}{\mathcal{L}}
\newcommand{\grad}{\nabla_{\pz,\pzinv,\pnz,\pext,\pnzinv,\lambda_1,\lambda_2,\lambda_3,\lambda_4}}
\newcommand{\srinv}{\sr^\text{inv}}
\newcommand{\sarinv}{\sar^\text{inv}}
\newcommand{\qfac}{\ensuremath{\bar{q}}}
\newcommand{\coeffpolyTheorem}{\ensuremath{\mathsf{c}}}
\newcommand{\sfb}{\ensuremath{\mathsf{b}}}
\newcommand{\sfd}{\ensuremath{\tilde{q}}}
\newcommand{\nrowsA}{\ensuremath{\mathsf{m}}}
\newcommand{\ncolsA}{\ensuremath{\mathsf{n}}}
\newcommand{\nrowsB}{\ensuremath{\mathsf{n}}}
\newcommand{\ncolsB}{\ensuremath{\mathsf{\ell}}}

\newcommand{\nbworkers}{\ensuremath{N}}

\newcommand{\polyanoind}[1][x]{\ensuremath{{f}(#1)}}
\newcommand{\polybnoind}[1][x]{\ensuremath{{g}(#1)}}
\newcommand{\respolynoind}[1][x]{\ensuremath{{h}(#1)}}
\newcommand{\polya}[1][x]{\ensuremath{{f_i}(#1)}}
\newcommand{\polyb}[1][x]{\ensuremath{{g_i}(#1)}}

\newcommand{\eval}[1][j]{\ensuremath{\alpha_{#1}}}
\newcommand{\shares}{\ensuremath{n}}
\newcommand{\coeff}{\ensuremath{\eval}}
\newcommand{\pc}{\ensuremath{p_s}}
\newcommand{\pcinv}{\ensuremath{\pc^\text{inv}}}
\newcommand{\totalleakagestragglers}{\mathrm{L}(\pz,\pzinv,\pc,\pcinv)}

\newcommand{\sd}{\ensuremath{s_d}}\newcommand{\sdinv}{\ensuremath{s_d^\text{inv}}}
\newcommand{\gradstragglers}{\nabla_{\pz,\pzinv,\pc,\pcinv,\lambda_1,\lambda_2,\lambda_3}}
\newcommand{\splitpsmm}{\ensuremath{m}}

\IEEEoverridecommandlockouts

\SetArgSty{textnormal}

\SetCommentSty{mycommfont}
\begin{document}
\newlength\figureheight
\newlength\figurewidth

\title{Sparsity and Privacy in Secret Sharing:\\ A Fundamental Trade-Off}

\author{\IEEEauthorblockN{
Rawad Bitar, \emph{Member, IEEE}}, Maximilian Egger, \emph{Student Member, IEEE}, Antonia Wachter-Zeh, \emph{Senior Member, IEEE}, and Marvin Xhemrishi, \emph{Student Member, IEEE}
\vspace{-.3cm}
\thanks{Preliminary results were presented in %
~\cite{sparse_ISIT,sparse_ITW,ISIT23}.} %

\thanks{The authors are with the Institute of Communications Engineering, School of Computation, Information and Technology, Technical University of Munich. Emails: \{rawad.bitar, maximilian.egger, antonia.wachter-zeh, marvin.xhemrishi\}@tum.de.}

\thanks{This work was partially funded by the DFG (German Research Foundation) under Grant Agreements No. WA 3907/7-1 and No. BI 2492/1-1.
}}
\maketitle
\begin{abstract}
This work investigates the design of sparse secret sharing schemes that encode a sparse private matrix into sparse shares. This investigation is motivated by distributed computing, where the multiplication of sparse and private matrices is moved from a computationally weak main node to untrusted worker machines. Classical secret-sharing schemes produce dense shares. However, sparsity can help speed up the computation. We show that, for matrices with i.i.d. entries, sparsity in the shares comes at a fundamental cost of weaker privacy. We derive a fundamental tradeoff between sparsity and privacy and construct optimal sparse secret sharing schemes that produce shares that leak the minimum amount of information for a desired sparsity of the shares. We apply our schemes to distributed sparse and private matrix multiplication schemes with no colluding workers while tolerating stragglers. For the setting of two non-communicating clusters of workers, we design a sparse one-time pad so that no private information is leaked to a cluster of untrusted and colluding workers, and the shares with bounded but non-zero leakage are assigned to a cluster of partially trusted workers. %
We conclude by discussing the necessity of using permutations %
for matrices with correlated entries.

\end{abstract}
\begin{IEEEkeywords} Sparse private matrix multiplication, straggler tolerance, optimal leakage,  information-theoretic privacy.
\end{IEEEkeywords}

\section{Introduction}\label{sec:intro}

Distributed computing became a ubiquitous requirement with the emergence of machine learning applications where an enormous amount of data is processed and, due to limitations of the computation power, forwarded to computing nodes~\cite{importance_distributed_comp}. We consider a {\master/{\worker}} setting, where a central entity called {\em \master} owns a large amount of private data and needs to run intensive computations on it. The intensive computation is split into multiple tasks of potentially smaller complexity that can be assigned to external computation nodes called {\em \workers} to be run in parallel.

Offloading computations to external nodes faces several challenges, such as leaking private information from sensitive data, which could be prohibited by data-protection laws~\cite{GDPR}, and increased overall latency due to straggler behavior, i.e., slow or unresponsive nodes \cite{Deanetal}.

We restrict our attention to distributed  matrix multiplication as being a computationally exhaustive part of several machine learning algorithms, such as support vector machines, principal component analysis, and gradient descent algorithms \cite{seber2012linear, suykens1999least}. A vast literature on the topic shows that coding-theoretic techniques can be used to mitigate the effect of stragglers and to ensure strong information-theoretic privacy of the data, e.g., \cite{Albin, GauriJoshi, Convolutional_codes, Bivarite_non_private, Anton_Alex, FLT_Codes, MatDot, speeding_up_using_codes, yang2018secure, kim2019private, RPM3, Kakar, GASP, LCC, raviv2019private, aliasgari2020private, makkonen2022general, PRAC, Staircase, reent, Burak_private} and the survey in \cite{ulukus2022private}.

However, those techniques %
require ``mixing'' the private data with random matrices (for privacy) to generate the tasks sent to the workers. Hence, all structure of the input matrices, and in particular sparsity, is lost. 
Sparse matrices, i.e., matrices with a relatively large number of zero entries, appear naturally in some machine learning tasks~\cite{madarkar2021sparse}. The structure of sparse matrices is important and can be leveraged by deploying efficient storage and computations strategies \cite{adaptive_sparse_matrix, implementing_sparse_matrix_vector}. 

In this work, we introduce the problem of encoding sparse private matrices into matrices that are still sparse and, when offloaded to workers, maintain information-theoretical privacy guarantees as will be formally defined in the sequel. We explain next the closest related works and highlight our contributions.

\emph{Related work: }Straggler mitigation in non-private distributed matrix multiplication is extensively studied in the literature, e.g.,~\cite{speeding_up_using_codes, Albin, GauriJoshi, Convolutional_codes, Bivarite_non_private, Anton_Alex, FLT_Codes, MatDot}. The idea is to encode the input matrices into new ones, called \emph{tasks}, such that stragglers can be mitigated when assigned to the workers. Sparsity-preserving encoding techniques in distributed matrix multiplication have recently received attention in the scientific community. The authors of~\cite{coded_sparse_mm} propose encoding the input matrices using an LT code~\cite{LT_codes} with a carefully tailored Soliton distribution to generate sparse matrices (tasks). %
In~\cite{coded_sparse_matrix_leverages_partial_stragglers}, the authors design encoding strategies for many settings, including distributing the multiplication of a sparse matrix with a dense vector. Fractional repetition codes are used to preserve a plausible level of sparsity in the tasks and mitigate stragglers. Moreover, an extra layer of non-sparse coded matrices is added to increase the straggler mitigation capability of the proposed codes. In~\cite{das2022unified}, the same authors improve their previous scheme~\cite{coded_sparse_matrix_leverages_partial_stragglers} by introducing a sparser encoding strategy of the last layer. In~\cite{das2023distributed}, the authors propose sparse encoding strategies that account for better numerical stability properties (compared to other existing schemes) and for heterogeneous task allocations (proportional to the {\workers} storage capabilities). In~\cite{RandomKhatriRaoProduct}, the authors propose to use a random Khatri-Rao product to construct a distributed (non-private) matrix-matrix multiplication that is numerically stable. The work~\cite{Krishna_sparse} modifies the construction of~\cite{RandomKhatriRaoProduct} to obtain sparse encoding strategies that have similar numerical stability when the input matrices are sparse. %
The works in~\cite{Lagrange_sparse_coding, Matrix_sparsification} consider the case where a dense input matrix is sparsified to speed up the overall computation.  

In terms of privacy, several encoding strategies that mitigate stragglers and guarantee perfect information-theoretic privacy have been proposed. The literature distinguishes between one-sided privacy, where the {\master} wants to offload the multiplication of two matrices, one of which has to remain private, e.g.,~\cite{reent,LCC, PRAC,Staircase}, and double-sided-privacy where both input matrices must remain private, e.g.,~\cite{yang2018secure, RPM3, GASP, Burak_private, Kakar, On_the_capacity, batch,aliasgari2020private, makkonen2022general}. However, none of the works on private distributed matrix multiplication considered encoding strategies that generate tasks that inherit the sparsity of the input matrices. The reason is that, as we shall show next, perfect information-theoretic privacy requires encoding strategies that completely destroy the sparsity of the input matrices.

Very recently, \cite{CK23} studied the tradeoff between privacy and storage rate in distributed storage.

\emph{Contributions: } We introduce the problem of constructing codes for private distributed matrix multiplication that mitigate stragglers and inherit sparsity from the private input matrices. We start by showing that insisting on perfect information-theoretic privacy does not allow leveraging the beneficial properties of sparse matrices. Therefore, we open the door to allowing sparsity and privacy to co-exist by relaxing the privacy guarantees. We prove the existence of a fundamental trade-off between sparsity and information-theoretic privacy. We optimize this trade-off to design optimal codes for private distributed matrix multiplication with sparsity guarantees.

\emph{Organization: } %
We set the notations in \cref{sec:prelim} and formulate the general problem in \cref{sec:system_model}. In~\cref{sec:sys_model}, we describe the considered system model and our main results. In~\cref{subsec:sparse_ot_pad}, we construct a novel sparse one-time padding strategy and optimize it for given requirements. \cref{sec:sss_polynomial} generalizes into a sparse secret-sharing scheme with an arbitrary number of shares. In~\cref{sec:double_sided}, we explain how to apply the introduced schemes to sparse and private matrix multiplication with stragglers. We discuss practical solutions for sparse matrices with correlated entries in~\cref{sec:discussion} %
and conclude the paper in section~\ref{sec:conclusion}. 

\section{Notation}\label{sec:prelim}
We use uppercase and lowercase bold letters to denote matrices and vectors, e.g., $\bfX$ and $\bfx$, respectively. By $\bfXij$, we denote the $(i,j)$-th entry of a matrix $\bfX$. We use letters in uppercase \emph{typewriter} font for random variables, e.g., $\rvY$. The random variables representing a matrix $\bfX$ and its $(i,j)$-th entry $\bfXij$ are denoted by $\rvX$ and $\Xij{i}{j}$, respectively. A finite field of cardinality $q$ is denoted by $\F_q$ and we use $\F_q^*$ to define its multiplicative group, i.e., $\F_q^* \triangleq \F_q\setminus \{0\}$. Sets are denoted by letters in calligraphic font, e.g., $\mathcal{X}$. For a positive integer $b$, we define the set $[b]\triangleq \{1,2,\dots, b\}$. Given $b$ random variables $\rvY_1,\dots,\rvY_b$ and a set $\mathcal{I}\subseteq [b]$, we denote by $\{\rvY_i\}_{i\in\mathcal{I}}$ the set of random variables indexed by $\mathcal{I}$, i.e., $\{\rvY_i\}_{i\in\mathcal{I}} \triangleq \{\rvY_{i} | i\in \mathcal{I}\}$. For two random variables $\rvX$ and $\rvY$, we denote $\Pr(\rvX = x)$ and $\Pr(\rvX = x\vert \rvY = y)$ by $\Pr_\rvX(x)$ and $\Pr_{\rvX\vert\rvY}(x\vert y)$, respectively. The probability mass function (PMF) of a random variable $\rvX \in \mathbb{F}_q$ is denoted by $\textrm{P}_\rvX = \left[p_1, p_2, \dots, p_q \right]$, i.e., $\Pr_{\rvX}(i) = p_i$ for all $i\in \mathbb{F}_q$. Throughout the paper, we will use $\textrm{H}_q(\rvX)$ and $\textrm{H}_q(\left[p_1, p_2, \dots, p_q \right])$ interchangeably to denote the $q$-ary entropy of the random variable $\rvX \in \mathbb{F}_q$. The $q$-ary mutual information between two random variables $\rvX$ and $\rvY$ is denoted by $\I_q(\rvX; \rvY)$. We will use $\kl{\textrm{P}_{\rvX}}{\textrm{P}_{\rvY}}$ to denote the Kullback-Leibler (KL) divergence between the PMFs of $\rvX$ and $\rvY$.

\section{Problem Formulation}\label{sec:system_model}
We introduce the problem of encoding private and sparse matrices into sparse matrices such that observing a collection of the output matrices reveals little to no information about the input matrices. In other words, the problem can be understood as constructing secret sharing schemes whose output is a collection of sparse matrices.

Formally, let $\bfA_1,\dots,\bfA_k$ be $k$ private matrices of the same dimensions drawn from a finite alphabet $\cA$. An $(n,\thresh,z)$ secret sharing (\sss) scheme with $\thresh = k+z$ is a randomized encoding $\sss: \cA^{k} \times \cA^z \to \cA^n$ that outputs $n$ matrices, called shares. Let $\rvA_1,\dots,\rvA_k$ and $\rvTask_1,\dots,\rvTask_n$ be random variables whose realizations are the private matrices and the output shares, respectively. A secret sharing scheme satisfies the following conditions: \begin{enumerate*}[label={\em \roman*)}] \item (reconstruction) any $\thresh$ shares suffice to reconstruct the input matrices, i.e., $\entropy(\rvA_1,\dots,\rvA_k|\rvTask_{i_1},\dots,\rvTask_{i_\thresh}) = 0$; and \item (perfect privacy) any $z$ out of the $n$ shares reveal no information about the private matrices, i.e., $\mutinf(\rvA_1,\dots,\rvA_k;\rvTask_{i_1},\dots,\rvTask_{i_z}) = 0$.
\end{enumerate*}

The problem with secret sharing schemes is that the shares are dense, i.e., have a relatively small number of zero entries, irrespective of the structure of the input matrices. Let $\sai$ be the sparsity\footnote{The sparsity could be understood in several ways, e.g., the ratio of the number of zero entries to the total number of entries, and will be defined formally when we explain the system model.} of the input matrices $\bfA_i$. Secret sharing schemes output shares with sparsity equal $1/|\cA|$ irrespective of the values of the $\sai\!$'s. 

We first investigate the existence of sparse secret sharing schemes (as defined next), which are of independent interest.

\begin{definition}[Sparse Secret Sharing]
\label{def:sss}
A sparse secret sharing scheme with parameters $(n,\thresh,z;\sleak,\spars_1,\dots,\spars_n)$ is a randomized encoding $\sparsess: \cA^k \times \cA^z \to \cA^n$ that takes as input $k=\thresh - z$ private and \emph{sparse} matrices with sparsity $\sai$, $i=1,\dots,k$, and outputs $n$ \emph{sparse} matrices called \emph{shares} that satisfy the following conditions:
\begin{enumerate}
    \item \emph{Reconstruction:} any $\thresh$ shares suffice to reconstruct the input matrices, i.e., $\entropy(\rvA_1,\dots,\rvA_k|\rvTask_{i_1},\dots,\rvTask_{i_\thresh}) = 0.$
    \item \emph{Privacy:} any $z$ out of the $n$ shares \emph{leak a small amount of information} about the private matrices, i.e., $$\mutinf(\rvA_1,\dots,\rvA_k;\rvTask_{i_1},\dots,\rvTask_{i_z}) \leq \sleak.$$
    \item \emph{Sparsity:} Each share $\task_i$ has a desired sparsity $\sasi = \spars_i$.%
\end{enumerate}
\end{definition}

The observant reader notices immediately that the privacy guarantee of sparse secret sharing schemes is relaxed from the perfect information-theoretic privacy, i.e., zero mutual information between the private matrices and any $z$ shares, to weak information-theoretic privacy that allows what we call an \emph{$\sleak$ leakage}. While perfect privacy is desirable, we shall show that there exists a tension between $\sleak$ and $\sasi$, $i = 1,\dots,n,$ and that if $\sleak = 0$ is desired, then $\sasi \leq 1/|\cA|$ which is satisfied by regular secret sharing schemes.

With sparse secret-sharing schemes, we define the problem of constructing straggler-tolerant sparse and private matrix multiplication schemes.

\begin{definition}[Sparse and Private Matrix Multiplication]
\label{def:spmm}
    A straggler-tolerant sparse and private matrix multiplication scheme with parameters $(\nbworkers,\stragglernum,z,\sleak_1,\sleak_2,\spars_1,\spars_2)$ takes as input two private and sparse matrices $\bfA$ and $\bfB$ and offloads their multiplication to $\nbworkers$ available worker nodes. The matrices are encoded each into $m\geq \nbworkers$ shares, %
    denoted by $\widetilde{\bfA}_1,\dots,\widetilde{\bfA}_m, \widetilde{\bfB}_1.\dots,\widetilde{\bfB}_m$. The shares are grouped into pairs called \emph{tasks} and denoted by $\bfT_{(i,j)} = (\widetilde{\bfA}_i,\widetilde{\bfB}_j)$, for $i,j\in[m]$, which are then assigned to the $\nbworkers$ workers. Denote by $\cW_\ell$ the set of all tasks assigned to worker $\ell$, $\ell \in [\nbworkers]$, i.e., $\cW_\ell \triangleq \{\bfT_{(i,j)} \mid \bfT_{(i,j)} \text{ assigned to worker } \ell\}$. Each worker computes sequentially the multiplication of all $(\widetilde{\bfA}_i,\widetilde{\bfB}_j) \in \cW_\ell$ and sends them back to the \master until either it finishes all its tasks or it is asked to stop computing. Let $\cC_\ell$ be the collection of multiplications sent by \worker $\ell$ to the \master. The scheme satisfies the following conditions:
        \begin{enumerate}
            \item \emph{Straggler tolerance:} The matrix $\bfC = \bfA\bfB$ can be recovered from the collection of any $\nbworkers-\sigma$ sets $\cC_\ell$, i.e., with a slight abuse of notation\footnote{For clarity of presentation, the set $\cC_\ell$ represents also a collection of random variables.} we can write for any distinct $i_1,\dots, i_{N-\sigma} \in \{1,\dots,N\}$
            \begin{equation*}
                \entropy(\rvC\mid \cC_{i_1},\dots,\cC_{i_{\nbworkers-\sigma}}) = 0.
            \end{equation*}
            \item \emph{Privacy:} Any $z$ workers can obtain a small amount of information about the input matrices, i.e.,
            \begin{align*}
                \mutinf(\rvA; \cW_{i_1},\dots,\cW_{i_z}) \leq \sleak_1, & &
                \mutinf(\rvB; \cW_{i_1},\dots,\cW_{i_z}) \leq \sleak_2,
            \end{align*}
            where $0 \leq \sleak_1 \leq \entropy(\rvA)$ and $0 \leq \sleak_2 \leq \entropy(\rvB)$
            \item \emph{Sparsity:} The average sparsity of the tasks in $\cW_\ell$ corresponding to matrices $\bfA$ and $\bfB$ are respectively equal to $\spars_1 \leq \sa$ and $\spars_2 \leq \spars_\bfB$.
        \end{enumerate}
\end{definition}

\section{System Model and Main Results}
\label{sec:sys_model}
\subsection{Sparsity and Private Matrices}

In this work, we restrict our attention to private matrices whose entries are independent and identically distributed.
\begin{definition}[Sparsity level of a matrix]
\label{def:iid} The sparsity level $\sparsity(\rvX)$ of a matrix $\rvX$ with entries independently and identically distributed is equal to the probability of the $(i,j)$-th entry $\rvX_{\{i,j\}}$ of $\rvX$ being equal to $0$, i.e., 
\begin{align*} 
\sparsity(\rvX) = \Pr_{\Xij{i}{j}}(0).
\end{align*}
\end{definition}

The private matrices of the {\master} are assumed to belong to a finite field $\bfA \in \F_q^{\nrowsA \times \ncolsA}$, $\bfB \in \F_q^{\nrowsB \times \ncolsB}$ and have sparsity levels $\sparsity(\bfA) = \sa$ and $\sparsity(\bfB) = \sbm$, where\footnote{For the special case $\sa,\sbm\leq q^{-1}$, the shares of classical secret sharing schemes have larger sparsity $q^{-1}$ and are, thus, enough.} $\sa, \sbm >q^{-1}$. 
The PMF of the entries of $\bfA$ (and $\bfB$) can be expressed as follows 
\begin{equation}
\label{eq:PMF_A}
\Pr_{\Aij}(a) = \begin{cases}
    \sa, & a=0 \\[1ex]
    \dfrac{1-\sa}{q-1}, & a \in \F_q^*
\end{cases}\,.
\end{equation} 

Though limiting, the assumption of the entries being i.i.d.\ serves as a stepping stone toward understanding the trade-off between privacy and sparsity. The goal is for the {\master} to offload the computation $\bfC = \bfA\bfB$ to $\nbworkers$ workers.

\subsection{Collusions and Stragglers}\label{subsec:no_clustering}
The {\workers} are assumed to be \emph{honest-but-curious} and will follow the protocol dictated by the \master. We say that up to $z$ \workers \emph{collude} if any collection of up to $z$ \workers collaborate to infer information about the private matrices from the tasks they have obtained. For the particular case of $z=1$, we say that the \workers do not collude.

\Workers are assigned multiple computational tasks that they compute and send their results sequentially to the \master. A \emph{full straggler} is a \worker that does not return any computation to the \master. A \emph{partial straggler} is a \worker that returns some of its computations to the \master while most of the other \workers have finished all their computations.

\subsection{Information Leakage}\label{Privacy_Measure}
Let the random variables $\rvA$ and $\rvB$ represent the private matrices $\bfA$ and $\bfB$, respectively, and $\cW_\ell$, $\ell \in [\nbworkers]$, denote the collection of random variables representing the tasks assigned to worker $\ell$. We say that the tasks assigned to worker $\ell$ leak $\sleak_\bfA \triangleq \mutinf(\rvA;\cW_\ell)$ amount of information about the matrix $\bfA$. We define the relative leakage by $\relleakage_\bfA \triangleq \mutinf(\rvA;\cW_\ell)/\entropy(\rvA)$, where $0\leq \relleakage_\bfA \leq 1$. %
Perfect information-theoretic privacy corresponds to $\sleak_\bfA = \relleakage_\bfA = 0$ and no privacy corresponds to $\sleak_\bfA = \entropy(\rvA)$ or $\relleakage_\bfA = 1$. The same can be defined for $\bfB$.

\begin{remark}[Operational meaning]
While the operational meaning of perfect information-theoretic privacy is clear: an eavesdropper with unbounded computational power cannot do better than guessing; a non-zero leakage has more ambiguity in its operational meaning. 
We are interested in understanding the fundamental theoretical trade-off between information-theoretic privacy and sparsity. 
In practice, a leakage of $\sleak>0$  means that any eavesdropper can learn some partial information about the private data. %
\end{remark}

\subsection{Main Results}\label{sec:DepSparseOTP}
We study sparse secret sharing schemes for $z=1$ and matrices with i.i.d entries. We start by showing in \cref{lemma:perfectprivacy} that a fundamental trade-off exists between sparsity and privacy. Namely, given a private matrix $\bfA$ with sparsity $\sa \geq {q}^{-1}$, a sparse secret sharing satisfies $\sasi > q^{-1}$ if and only if $\sleak >0$.

\begin{lemma} \label{lemma:perfectprivacy}
    Let $\rvA$ represent a sparse matrix that has i.i.d entries. For any secret sharing scheme with $z=1$, perfect information-theoretic privacy can be achieved if and only if the padding matrix $\rvR$ is generated independently of $\rvA$ with its entries being i.i.d. uniformly distributed. 
\end{lemma}
\begin{proof}
     The proof is shown in~\cref{proof1}.
\end{proof}

Therefore, to provide sparsity guarantees of the encoded matrices, we have to deviate from perfect information-theoretic privacy. We start by constructing a sparse one-time pad, cf. \cref{constr:sparse_one_time pad}, that takes as input a private sparse matrix $\bfA$ and outputs two sparse shares $\bfR$ and $\bfA+\bfR$ each of which leaks a pre-specified amount of information about $\bfA$. We show the existence of a fundamental tradeoff between the achievable sparsity of the shares and the privacy guarantees. Further, given a sparse matrix $\bfA$ and fixing the same desired sparsity of the shares, \cref{constr:sparse_one_time pad} provides shares with the provably smallest leakage, cf. \cref{lemma:optimization,thm:optimal_pmf}. We focus on shares with the same sparsity levels since by \cref{lemma:optimaldelta} we know that for a desired average sparsity, constructing shares with the same sparsity levels minimizes leakage. We remark in \cref{lemma:sparse_perfect_ot_pad} the existence of sparse one-time pads where one share leaks no information about $\bfA$ and characterize the sparsity-privacy tradeoff in such a case.

In \cref{def:model}, we construct sparse secret sharing schemes with $t=2$, $z=1$ and $\shares\geq 2$. We observe that the sparsity-privacy tradeoff extends to include the number of shares, i.e., for a fixed desired sparsity level, increasing the number of shares decreases the achievable privacy guarantees. \cref{thm:optimal_pmf_straggler_tol} shows that for the same desirable sparsity of the shares, \cref{def:model} provides the provably smallest leakage possible.

We apply the designed sparse secret sharing schemes to construct sparse and private matrix multiplication schemes with parameters $(\nbworkers, \sigma = \nbworkers -3, \sleak_1^\star, \sleak_2^\star, \spars_{d_1}, \spars_{d_2})$, where $\sleak_1^\star$ and $\sleak_2^\star$ are the optimal leakages obtained for the desired sparsity levels $\spars_{d_1}$ and $\spars_{d_2}$. We show how to use this construction as a building block for different schemes with larger straggler tolerance and smaller computation tasks per \worker, cf. in \cref{sec:double_sided}. 

Furthermore, we consider the setting in which the \workers can be grouped into two non-communicating clusters. One cluster consists of fully untrusted \workers, while the \workers of the other cluster are partially trusted. In this setting, we construct straggler-tolerant sparse and private matrix multiplication schemes that can tolerate collusions. We apply the result of \cref{lemma:sparse_perfect_ot_pad} to guarantee perfect information-theoretic privacy against all the untrusted \workers and guarantee a desired leakage for any $z$ partially trusted \workers, cf. \cref{thm:main}.

We provide a discussion on the effect of correlation between the entries of the matrix $\bfA$ (and $\bfB$) on the constructed schemes and show how to break the correlation with random permutations, cf. \cref{sec:discussion}.

\section{Sparse one-time pad}
\label{subsec:sparse_ot_pad}

To build intuition, in this section, we construct a sparse one-time pad, i.e., a secret sharing scheme with parameters $n=2, t = 2, z =1$, and we require $q^{-1} < s_1, s_2 \leq s$ (cf.~\cref{def:sss}). The goal is to understand how a sparse secret sharing scheme can be constructed and to quantify the optimal leakage $\sleak$ that can be obtained for fixed sparsity levels $\spars_1,\spars_2$. 

A one-time pad~\cite{shannon_one_time_pad} takes as input the matrix $\bfA \in \F_q^{\nrowsA \times \ncolsA}$, generates a matrix $\bfR \in \F_q^{\nrowsA \times \ncolsA}$ whose entries are generated i.i.d uniformly at random from $\F_q$ and outputs $\bfR$ and $\bfA+\bfR$ as shares. The key ingredient to obtaining a \emph{sparse} one-time pad is to generate the entries of $\bfR$ \emph{dependently} on $\bfA$. %
We consider generating $\bfR$ by using a conditional PMF as in~\cref{constr:sparse_one_time pad}.
\begin{construction}[Generation of the padding matrix $\bfR$]\label{constr:sparse_one_time pad}
To construct a sparse one-time pad, the entries of the padding matrix $\bfR$ are generated dependently on the entries of $\bfA$ as\footnote{Similarly, one can define a PMF that generates a padding matrix $\bfS$ depending on $\bfB$.}
\begin{align}
 \label{eq:dependent_on_0}
    \Pr_{\Rij\lvert \Aij}(r\lvert 0) \!&\!= \begin{cases} 
    \pz, &r = 0 \\   \pzinv , &r \neq 0,
    \end{cases}\,,\\[1ex]
\label{eq:dependent_on_nz}
    \Pr_{\Rij\lvert \Aij}(r\lvert a) \!&\!= \begin{cases} 
    \pext, &r = 0 \\ \pnz, &r = -a \\ \pnzinv , &r \not\in \{0,-a\},
    \end{cases}\,,
\end{align}
where $\pzinv \triangleq \frac{1-\pz}{q-1}$ and $\pnzinv \triangleq \frac{1-\pext-\pnz}{q-2}$. 
\end{construction}
Intuitively, the PMFs presented in~\cref{eq:dependent_on_0,eq:dependent_on_nz} ensure sparsity and privacy as follows. The event occurring with probability $\pz$ ensures that a zero entry in the matrix $\bfA$ is inherited in both $\bfA + \bfR$ and $\bfR$. The events occurring with probabilities $\pext$ and $\pnz$ create a zero entry in $\bfR$ and $\bfA+\bfR$, respectively, where there was a non-zero entry in $\bfA$. The events with probabilities $\pzinv$ and $\pnzinv$ create non-zero entries in both $\bfR$ and $\bfA+\bfR$ irrespective of the value of the corresponding entry of $\bfA$. We do not distinguish cases between different values of the non-zero entries of $\bfA$ since we assume that they are uniformly distributed over $\F_q\setminus\{0\}$. We numerically validate our construction by observing that a sparse one-time pad with desired sparsity levels attains the minimum leakage possible when considering matrices $\bfR$ as in~\cref{eq:dependent_on_0,eq:dependent_on_nz}.

\subsection{Analysis of Sparsity and Leakage}

\cref{lemma:sparsity} shows the effect of $\pz$, $\pext$ and $\pnz$ on the sparsity of $\bfR$ and $\bfA+\bfR$. The increase of $\pz$ increases both sparsity levels $\sr$ and $\sar$. On the other hand, the increase of $\pext$ and $\pnz$ increases only $\sr$ and $\sar$, respectively.

\begin{proposition}\label{lemma:sparsity}
Given an input matrix $\bfA$ with sparsity level $\sa$ whose entries are i.i.d, and if a padding matrix $\bfR$ is constructed using the conditional PMFs as in~\cref{eq:dependent_on_0,eq:dependent_on_nz}, then the sparsity levels $\sr$ and $\sar$ are 
\begin{align}
    \sr &= \pz \sa + \pext (1-\sa), \label{eq:sr}\\
    \sar &= \pz \sa + \pnz(1-\sa). \label{eq:sar}
\end{align}
\end{proposition}
\begin{proof}
    The proof follows by calculating the PMF of $\rvR$ and $\rvA+\rvR$ and applying the definition of sparsity.
\end{proof}

At first glance, it may seem that increasing $\pz$, $\pext$, and $\pnz$ is beneficial. However, increasing those parameters will also increase the leakage of $\bfR$ and $\bfA+\bfR$. Observe that choosing $\pz = \pext = 1$ and $\pnz = 0$ results in the highest sparsity possible of $\bfR$ equal to one and allow $\bfA+\bfR$ to have the same sparsity level of $\bfA$. However, $\bfR$ is the all-zero matrix. Despite $\bfR$ leaking no information about $\bfA$, observing $\bfA+\bfR = \bfA$ reveals all the entries of $\bfA$. To understand this tension between privacy and sparsity, we quantify the minimum leakage that can be obtained from both $\bfR$ and $\bfA+\bfR$ for fixed $\spars_\bfR$ and $\spars_{\bfA+\bfR}$.

Define $\leakage{\bfR}\triangleq\mutinf(\rvR;\rvA)$ and $\leakage{{\bfA+\bfR}}\triangleq\mutinf(\rvA+\rvR;\rvA)$ to be the leakage of $\bfR$ and $\bfA+\bfR$, respectively. Minimizing the leakage amounts to minimizing those quantities. Since the entries of $\bfA$ are assumed to be i.i.d. and the entries of $\bfR$ are also i.i.d, we can write $\leakage{\bfR} = \nrowsA\ncolsA \leakage{1} \triangleq \nrowsA\ncolsA \mutinf(\rvR_{\{i,j\}};\rvA_{\{i,j\}})$ and $\leakage{{\bfA+\bfR}} = \nrowsA\ncolsA \leakage{2} \triangleq \nrowsA\ncolsA \mutinf(\rvA_{\{i,j\}}+\rvR_{\{i,j\}};\rvA_{\{i,j\}})$. It then suffices to minimize $\leakage{1}$ and $\leakage{2}$. Our goal is to minimize the total leakage of a sparse one-time pad defined as $\leakage{\text{tot}} \triangleq\leakage{1}+\leakage{2}$. %
Minimizing the leakage is done over all possible choices of conditional PMFs $\mathrm{P}_{\Rij\lvert\Aij}$, i.e., over the set of all possible values of $\cP \triangleq \{\pra : r,a \in \F_q\}$. Therefore, minimizing the leakage is equivalent to the following minimization problem\footnote{For ease of presentation, we omit the subscripts $\{i,j\}$ in the second to last equality.}
\begin{align*}
    \lossopt &= \optimizer \leakage{\text{tot}} =\optimizer \leakage{1} + \leakage{2} \\
    &= \optimizer \mutinf(\rvR_{\{i,j\}};\rvA_{\{i,j\}}) + \mutinf(\rvA_{\{i,j\}}+\rvR_{\{i,j\}};\rvA_{\{i,j\}}) \\
    &=\optimizer \begin{aligned}[t] &\kl{\paandr}{\pa \pr} + \kl{\paandapr}{\pa \papr} \end{aligned} \\
    &=\optimizer \!\!\! \sum_{a,b\in\mathbf{F}_q} \!\!\! \pa(a) \!\left( \pra[b-a][a] \log\frac{\pra[b-a][a]}{\papr(b)} \!+\! \pra[b][a] \log\frac{\pra[b][a]}{\pr(b)} \right)\!.
\end{align*}
This is a constrained optimization problem whose constraints follow from the requirement of having a valid PMF and the desired sparsities $\sr$ and $\sar$, and can be written as 
\begin{align*}
    \forall a\in\Fq: \pra[0][a] + \sum_{r\in\Fqstar} \pra[r][a] - 1 &= 0, \\[-5pt]
    \pra[0][0] \cdot \sa + \sum_{a\in\Fqstar} \pra[0][a] \cdot \pa(a) -\sr &= 0, \\[-5pt]
    \pra[0][0] \cdot \sa + \sum_{a\in\Fqstar} \pra[-a][a] \cdot \pa(a) -\sar &= 0.
\end{align*}
Optimizing over $\cP$ means that the optimization problem should be performed over $q^2$ variables. From numerical simulations, we observed that when the entries of $\bfA$ are distributed as in~\cref{eq:PMF_A}, then the optimization boils down to optimize only over the three variables $\pz,\pext$ and $\pnz$ shown in~\cref{constr:sparse_one_time pad}, i.e., the probabilities that do not contribute into the sparsity levels $\sr,\sar$ can be uniformly distributed. Then, the optimization problem becomes as shown next.

\begin{lemma}
\label{lemma:optimization}
If the entries of the private matrix $\bfA$ are distributed as in~\cref{eq:PMF_A} and the desired sparsity levels $\sr, \sar$, of the padding and padded matrices, are fixed, then the optimal element-wise total leakage is given by optimizing $\min\limits_{\pz, \pext, \pnz} \totalleakage$, where
\begin{align*}
&\totalleakage = s \bigg[ \pz \bigg( \log \frac{\pz}{\sar} + \log \frac{\pz}{\sr} \bigg) \nonumber\\
&+(q-1)\pzinv \bigg( \log \frac{\pzinv}{\sarinv} + \log \frac{\pzinv}{\srinv}\bigg) \bigg] + (1-s) \nonumber\\
&\cdot \bigg[ \pext \bigg( \log \frac{\pext}{\sarinv} + \log \frac{\pext}{\sr} \bigg) + \pnz \bigg( \log \frac{\pnz}{\sar} + \log \frac{\pnz}{\srinv} \bigg) \nonumber\\ &+(q-2) \pnzinv \bigg( \log \frac{\pnzinv}{\sarinv} + \log \frac{\pnzinv}{\srinv} \bigg)
\bigg],
\end{align*}
for $\srinv \triangleq (1-\sr)/(q-1)$, $\sarinv \triangleq (1-\sar)/(q-1)$ with the following constraints
\begin{align}
    c_1(\pz,\pzinv) &\triangleq \pz + (q-1) \pzinv - 1 &&= 0, \label{eq:constraint1} \\
    c_2(\pext,\pnz,\pnzinv) &\triangleq \pext + \pnz + (q-2) \pnzinv - 1 &&= 0, \label{eq:constraint2} \\
    c_3(\pz,\pext) &\triangleq \pz s + \pext (1-s) - \sr &&= 0, \label{eq:constraint3} \\
    c_4(\pz,\pnz) &\triangleq \pz s + \pnz (1-s) - \sar &&= 0. \label{eq:constraint4}
\end{align}
\end{lemma}
\begin{proof}
The proof follows by using the definition of KL-divergence and is omitted for brevity.
\end{proof}
For analytical tractability, we further simplify the problem by fixing a desired average sparsity level $\savg \triangleq \frac{\sr+\sar}{2}$ instead of fixing the values of $\sr$ and $\sar$ separately. We show in \cref{lemma:optimaldelta} (\cref{proof5}) that for a fixed $\savg$, the minimum total leakage is obtained for $\sr = \sar$.
Therefore, in the sequel, we will consider minimizing the total leakage for a desired average sparsity $\savg$, where the sparsity levels of the matrices are $\sr = \sar = \savg$. 

\subsection{Optimal Sparse One-Time Pad}

We show in \cref{thm:optimal_pmf} that the optimization problem in \cref{lemma:optimization} is a convex optimization problem that admits only one solution. Hence, the result is twofold: \begin{enumerate*}[label={\em (\roman*)}] \item given a desired $\savg$, the minimum leakage possible is characterized; and \item by constructing the matrix $\bfR$ as in \cref{constr:sparse_one_time pad} with parameters $\pz^\star,\pext^\star$ and $\pnz^\star$ as in \cref{thm:optimal_pmf}, the desired sparsity and the minimum leakage possible are attained.
\end{enumerate*}
\begin{theorem} \label{thm:optimal_pmf}
Given a desired sparsity level $\savg$, the optimal PMF of the form given in \cref{eq:dependent_on_0,eq:dependent_on_nz} that minimizes the leakage is obtained by picking $\pz=\pz^\star$ as the root of a polynomial $\coeffpolyTheorem_3 \pz^3 + \coeffpolyTheorem_2 \pz^2 + \coeffpolyTheorem_1 \pz + \coeffpolyTheorem_0$ that satisfies $\max\{2\savg-1+s,0\} \leq 2 \pz s \leq 2\min\{s,\savg\}$. The coefficients of the polynomial are 
\begin{align*}
\coeffpolyTheorem_3 &\triangleq s^2 (4+\qfac), \\
\coeffpolyTheorem_2 & \triangleq 4s(1-s-2\savg)-\qfac s (2\savg+s),\\
\coeffpolyTheorem_1 & \triangleq (1-s-2\savg)^2+\qfac(2s\savg + \savg^2),\\
\coeffpolyTheorem_0 &\triangleq -\qfac \savg^2,
\end{align*}
where $\qfac \triangleq (q-2)^2/(q-1)$. The parameters $\pzinv,\pext^\star,\pnz^\star$ can be obtained by using $\pz^\star$ in \cref{eq:constraint1,eq:constraint2,eq:constraint3,eq:constraint4}.
\end{theorem}
\begin{proof}
    The proof is shown in~\cref{proof6}.
\end{proof}

\subsection{Semi-Perfect One-Time Pad}
\label{subsec:sparse_one_time_pad}

\cref{lemma:perfectprivacy} closes the door on obtaining a perfectly secure (with zero leakage) one-time pad with sparse shares. 
More precisely, given a sparse matrix $\bfA$, if $\bfR$ is generated as in~\cref{constr:sparse_one_time pad},  then $\leakage{1} + \leakage{2} =0$ occurs if and only if $\sr = \sar = q^{-1}$. Clearly, if the field size $q$ is very large, then the shares cannot be considered sparse. However, \cref{lemma:perfectprivacy} does not rule out the case where $\leakage{1} + \leakage{2} >0$ and $\leakage{1}\cdot\leakage{2} =0$, i.e., one of the shares achieves a zero leakage. %
However, constructing a semi-perfect one-time pad scheme can be useful in distributed matrix multiplications where part of the \workers are fully untrusted while the other part can be partially trusted, cf.~\cref{sec:sp_per_it_priv}.
We show that it is possible to construct shares $\bfR$ and $\bfA+\bfR$ where the latter achieves perfect privacy, hence allowing for a sparse semi-perfect one-time pad. This can be achieved using~\cref{constr:sparse_one_time pad} with $\pz = \pnz = p$ and $\pext = \frac{1-p}{q-1}$, for any $0<p<1$ as shown in~\cref{lemma:sparse_perfect_ot_pad}.  

\begin{lemma}
\label{lemma:sparse_perfect_ot_pad}
Given a private sparse matrix $\bfA$ with sparsity level $\sa = s > q^{-1}$ and if a padding matrix $\bfR$ is generated dependently on $\bfA$ as in~\cref{constr:sparse_one_time pad} 
for $\pz = \pnz =p$ and $\pext = \frac{1-p}{q-1}$ for $0<p<1$, then the sparsity levels of the padding and the padded matrix $\sr$ and $\sar$ and the leakage $\leakage{1}$ increase with $p$, while no private information about $\bfA$ is leaked through the padded matrix, i.e., $\leakage{2} =0$.
\end{lemma}
\begin{proof}
The sparsity levels of the shares are computed by plugging $\pz = \pnz = p$ and $\pext = \frac{1-p}{q-1}$ into~\cref{eq:dependent_on_0,eq:dependent_on_nz} to obtain 
\begin{align*}
    \sr  = p\frac{(sq-1)}{q-1} + \frac{(1-s)}{q-1}, \text{ and }\qquad 
    \sar = p\,.
\end{align*}
Clearly, $\sar$ increases with $p$ and $\sr$ increases with $p$ for the case where $s>q^{-1}$. 

From~\cite[Theorem 2.7.4]{cover2012elements}, we conclude that the leakage $\leakage{1} = \mutinf({\Rij;\Aij})$ is a convex function in $p$ since the PMF of the input matrix is fixed and the conditional PMF of $\Rij\lvert \Aij$ can be written as a convex mixture of two conditional distributions as follows
    \begin{align*}
    \Pr_{\Rij \lvert\Aij}(r\lvert a) = p \mathbbm{1}_{r = -a} + (1-p)\frac{1}{1-q} \mathbbm{1}_{r \neq -a}\,,
    \end{align*}
$\forall r,q \in \F_q$, where $\mathbbm{1}_{\text{condition}}$ is the indicator function that returns $1$ if the condition is true and $0$ otherwise (in this case the indicator functions are  conditional distributions). 

The leakage $\leakage{2} = \mutinf({\Aij+\Rij;\Aij})$ is shown to be equal to zero by writing down the definition of mutual information and noting that the entropy is independent of the alphabet and only depends on the PMF.
\end{proof}

\cref{lemma:sparse_perfect_ot_pad} introduces a trade-off between the sparsity levels of the shares $\sa, \sar$ and the leakage $\leakage{1}$, which we will optimize in the context of sparse and private matrix multiplication schemes that can tolerate stragglers (\cref{sec:sp_per_it_priv}).

\section{Sparse Secret Sharing}
\label{sec:sss_polynomial}
In this section, we construct sparse secret sharing schemes for parameters $\shares\geq 3$, $t = 2$ and $z= t-1=1$, i.e., a generalization of the well-known Shamir secret sharing schemes~\cite{S79} for $t=2$. Given an input matrix $\bfA$, the challenge is to design a random matrix $\bfR$ such that $\shares$ evaluations of the encoding polynomial  
\begin{align*}
    \polyanoind = \bfA + x \bfR\, ,
\end{align*}
give $\shares$ shares that satisfy the desired sparsity and privacy guarantees. Our approach is to design the matrix $\bfR$ based on the entries of $\bfA$ and the evaluation points $\coeff[1],\dots,\coeff[\shares]$ used to generate the shares as shown next.

\begin{construction}[Conditional distribution of the random matrix for sparse polynomial encoding]\label{def:model} Given a set of $\shares$ distinct non-zero evaluation points $\coeff[i]$ for $i \in [\shares]$, the entries of the matrix $\bfR$ are drawn according to the following conditional distribution.
\begin{align}
    \Pr_{\Rij \lvert \Aij}(r\lvert 0) \!&\!= \begin{cases} 
    \pz, &r = 0 \\   \frac{1-\pz}{q-1} , &r \neq 0,
    \end{cases} \label{eq:gdependent_on_0_final}\\
    \Pr_{\Rij\lvert \Aij}(r\lvert a) \!&\!= \begin{cases} 
    \pc, &r \in \{-\frac{a}{\coeff[i]}\}_{i \in [\shares]} \\
    \frac{1-\shares \pc}{q-\shares} , &r\not\in \{-\frac{a}{\coeff[i]}\}_{i \in [\shares]}, \\
    \end{cases} \label{eq:gdependent_on_nz_final}
\end{align}

where $r\in \F_q$, $a \in \F_q^*$, $\shares$ is the number of shares and $\frac{1}{\coeff[i]}$ is the multiplicative inverse of $\coeff[i]$ in $\mathbb{F}_q$. 
\end{construction}

Similarly to~\cref{constr:sparse_one_time pad}, 
the intuition behind~\cref{def:model} is to pick probabilities that contribute to having zero-valued entries in the shares while reducing the dependence between the shares and $\bfA$ to reduce the leakage. The parameter $\pc$ tailors the probability of artificially having zero-valued entries in a single share, while $\pz$ controls the probability of inheriting a $0$ in every share. %
\cref{def:model} is restricted to generating $\shares$ shares that have the same sparsity levels. The motivation of this restriction is that for $\shares = 2$, the total leakage is minimized when all shares have the same sparsity, cf. \cref{lemma:optimaldelta}. Therefore, we conjecture that this holds for $\shares\geq 3$, i.e., if more than $2$ shares are generated per input, then the optimal leakage is obtained when all shares have the same sparsity levels. 

\cref{lemma:sparsity_polys} quantifies the sparsity level of shares constructed as in \cref{def:model}, \cref{lemma:leakage_straggler_tol} presents the optimization problem and \cref{thm:optimal_pmf_straggler_tol} provides the values of $\pz$ and $\pc$ that minimize the total leakage.
\begin{lemma}
\label{lemma:sparsity_polys}
Following \cref{def:model}, the shares have the following level of sparsity:
\begin{align*}
    \sd &= \pz \sa + \pc (1-\sa)\,.
\end{align*}
\end{lemma}
\begin{proof}
    The proof follows by calculating the PMF of the shares and applying the definition of sparsity.
\end{proof}

\newcommand{\xlogx}[2]{\ensuremath{z(#1, #2)}}
We define $\xlogx{x}{y} \triangleq x\log(\frac{x}{y})$ for ease of notation in the remainder of the paper.
\begin{lemma}
\label{lemma:leakage_straggler_tol}
Let the entries of the private matrix $\bfA$ be distributed as in~\cref{eq:PMF_A}. When considering shares of the form given in \cref{def:model} with fixed desired sparsity level $\sd$, the optimal element-wise total leakage is given by $\min\limits_{\pz, \pc} \totalleakagestragglers$, where
\begin{align}
&\totalleakagestragglers = s \big[ \xlogx{\pz}{\sd} +(q-1)\xlogx{\pzinv}{\sdinv} \big] \nonumber \\
& \!+ \! (1-s) \big[ \xlogx{\pc}{\sd} \!+\! (\shares-1) \xlogx{\pc}{\sdinv} \!+\! (q-\shares) \xlogx{\pcinv}{\sdinv} \big]\!. \label{eq:leakage_straggler_tol}
\end{align}
Thereby, $\sdinv = \frac{1-\sd}{q-1}$ is the likelihood of a share's entry being non-zero. The constraints of the optimization problem are
\begin{align}
    c_1(\pz,\pzinv) &\triangleq \pz + (q-1) \pzinv - 1 &&= 0, \label{eq:gconstraint1strag} \\
    c_2(\pc,\pcinv) &\triangleq \shares \pc + (q-\shares) \pcinv - 1 &&= 0, \label{eq:gconstraint2strag} \\
    c_3(\pz,\pc) &\triangleq \pz s + \pc (1-s) - \sd &&= 0\,.
    \label{eq:gconstraint3strag}
\end{align}
Further, $\pz,\pzinv, \pc$ and $\pcinv \triangleq \nicefrac{(1-\shares \pc)}{(q-\shares)}$ are bounded between zero and $1$.
\end{lemma}
\begin{proof}
For each share's fixed desired sparsity level $\sd$, minimizing the total leakage for \emph{any sparse secret sharing scheme} amounts to solving the following optimization problem.

\begin{align}
    \lossopt &= %
    \optimizer \sum_{i \in [\shares]} \mutinf \left(\Aij+\coeff[i]\Rij; \Aij\right) \nonumber\\
    &=\optimizer \!\!\! \sum_{i \in [\shares]} \sum_{a,y\in\mathbb{F}_q} \!\!\! \pb(a) \!\left(\pra[(y-a){\coeff[i]}^{-1}][a] \log\frac{\pra[(y-a){\coeff[i]}^{-1}][a]}{\mathrm{P}_{\rvA+\coeff[i]\rvR}} \right)\!.\label{eq:matrix_optimization}
\end{align}

Again, this is a constrained optimization problem with the following constraints %
\begin{align*}
    \forall b\in\Fq: \pra[0][b] + \sum_{x\in\Fqstar} \pra[x][b] - 1 &= 0, \\
    \!\forall i \in [\shares]: \pra[0][0] \cdot \pb(0) + \sum_{b\in\Fqstar} \pra[-b\coeff^{-1}][b] \cdot \pb(b) -\sd &= 0. \! %
\end{align*}
When restricting our attention to shares of the form given in \cref{def:model}, the optimization problem transforms into the one given in the statement of the Lemma.
\end{proof}

The solution to the optimization problem in \cref{lemma:leakage_straggler_tol} is given in the following \cref{thm:optimal_pmf_straggler_tol}.

\begin{theorem} \label{thm:optimal_pmf_straggler_tol}
Given a desired sparsity level $\sd$, the value $\pc^\star$ that minimizes the leakage of the shares as in \cref{def:model} is the real root of the polynomial $\sum_{j = 0}^{\shares+1} \sfb_j \pc^j$ with coefficients 
\begin{align*}
    &\sfb_{\shares+1} = -1 - (-\shares)^{\shares}\sfd \\
    &\sfb_{\shares} = \big( \tilde{s} (-\shares)^\shares - \shares (-\shares)^{\shares-1} \big)\sfd -\bar{s} \\
    &\sfb_k = \left(\tilde{s} \binom{\shares}{k} (-\shares)^k - \binom{\shares}{k-1} (-\shares)^{k-1} \right)\sfd, \forall k\in [\shares-1] \\
    &\sfb_0 = \sfd\tilde{s},
\end{align*}
where $\tilde{s} = \nicefrac{\sd}{(1-s)}$, $\bar{s}= \nicefrac{(s-\sd)}{(1-s)}$ and $\sfd = \nicefrac{(q-1)}{(q-\shares)^\shares}$. The root $\pc^\star$ must satisfy $0 \leq \pc^\star (1-s) \leq \min\{\sd, \frac{1}{n} \}$ and $\pz^\star$ can be computed as
\begin{equation*}
    \pz^\star = \frac{\sd - \pc^\star(1-s)}{s}.
\end{equation*}
\end{theorem}

\begin{proof}

The proof is provided in~\cref{proof8}.
\end{proof}

\begin{figure}[t]
    \centering
    \resizebox{.95\linewidth}{!}{
    \input{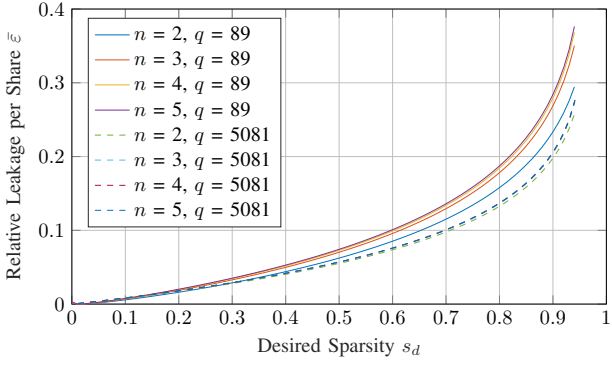}
    }
    \caption{The relative leakage of $\polyb[\eval]$ versus its desired sparsity for different number of shares $\shares$, $s=0.95$ and $q \in \{89, 5081\}$.}
\label{fig:leakage_over_sparsity}
\end{figure}

\cref{fig:leakage_over_sparsity} depicts the relative leakage $\relleakage$ per share as a function of $\sd$ and the number of shares $\shares$. We consider the finite field $\Fq$ for $q=89$ and $q=5081$. As $\shares$ grows\footnote{This statement holds for $\shares \ll q$. If $\shares$ approaches $q$, the constraints imposed by the optimization problem in \cref{eq:matrix_optimization} lead to non-sparse secret sharing.}, each share leaks more about the private input matrix for a fixed $\sd$. For large values of $q$, the loss of privacy incurred by increasing $\shares$ becomes negligible. %
In addition, for fixed $\shares$ and fixed $\sd$, increasing $q$ decreases the relative leakage per share $\relleakage$. For example, for $q=89$ and $\sd = 0.9$, $\relleakage$ is equal to $0.234$ for $\shares = 2$ and increases to $0.284$ when increasing $\shares$ to $5$. However, for $q=5081$ and $\sd=0.9$, $\relleakage$ increases from $0.199$ to $0.207$ when increasing $\shares$ from $2$ to $5$. Notice that for fixed $\shares$ and $\sd$, increasing $q$ decreases $\relleakage$, cf. \cref{fig:leakage_over_sparsity}.

We show next how to apply the designed sparse secret sharing schemes to distributed matrix multiplication.

\section{Sparse and Private Matrix Multiplication Schemes}\label{sec:double_sided}
In this section, we focus on designing straggler tolerant sparse and private matrix multiplication schemes.

We construct schemes with parameters $(\nbworkers, \sigma = \nbworkers -3, \sleak_1^\star, \sleak_2^\star, \spars_{d_1}, \spars_{d_2})$, see~\cref{def:spmm}, where $\sleak_1^\star$ and $\sleak_2^\star$ are the optimal leakages obtained for the desired sparsity levels $\spars_{d_1}$ and $\spars_{d_2}$. We then tackle the setting in which the matrix $\bfB$ is public and the workers can be grouped into two clusters: fully non-trusted workers to whom no information about $\bfA$ should be leaked; and partly trusted workers to whom a small positive leakage about $\bfA$ can be tolerated. In this setting, we construct a scheme that satisfies the perfect privacy of $\bfA$ for one cluster and minimizes the leakage for the other. The constructed scheme tolerates a desired number of full and partial stragglers.

\subsection{General Setting} 
\label{subsec:sp_pr_polynomial}

\begin{figure*}[!t]
    \centering
    \resizebox{.95\textwidth}{!}{\begin{tikzpicture}[>=stealth', auto,
		triangle/.style = {fill=white, regular polygon, regular polygon sides=3 }]
\definecolor{blue}{rgb}{0.19, 0.55, 0.91}
\definecolor{plum}{rgb}{0.8, 0.6, 0.8}
\def\blue{\color{blue}}
\def\ogreen{\color{ogreen}}
\def\orange{\color{orange}}
\def\nred{\color{plum}}
\definecolor{lightgray}{rgb}{0.83, 0.83, 0.83}
\def\mx{0.4}
\def\my{1}
\tikzstyle{matrixa} = [rectangle, rounded corners = 2mm, draw = lightgray, minimum height = 2.2cm, minimum width = 1.5 cm, fill = blue!40]
\tikzstyle{matrixr} = [rectangle, rounded corners = 2mm, draw = lightgray, minimum height = 2.2cm, minimum width = 1.5 cm, fill = blue!20]
\tikzstyle{matrixa1} = [rectangle, rounded corners = 1mm, draw = lightgray, minimum height = 0.75 cm, minimum width = 1.5 cm, fill = blue!40]
\tikzstyle{matrixr1} = [rectangle, rounded corners = 2mm, draw = lightgray, minimum height = 0.75cm, minimum width = 1.5 cm, fill = blue!20]			
\tikzstyle{matrixa2} = [rectangle, rounded corners = 1mm, draw = lightgray, minimum height = 0.75cm, minimum width = 1.5 cm, fill = blue!40]
\tikzstyle{matrixb} = [rounded corners = 2mm, draw = lightgray, minimum height = 1.5cm, minimum width = 0.3 cm, fill = ogreen!30]	
\tikzstyle{tasks1} = [rounded corners = 2mm, draw = lightgray, minimum height = 0.75cm, minimum width = 2.5 cm, fill = plum!50]		
\tikzstyle{tasks2} = [rounded corners = 2mm, draw = lightgray, minimum height = 0.75cm, minimum width = 2.5 cm, fill = blue!20]		
\tikzstyle{server} = [fill=black!10, rectangle, rounded corners=4mm, draw,minimum width=2em, minimum height=2.5em]
\node[inner sep=0] (s1) at (-1,0) {};
\node[inner sep=0pt, below left = 0.6 and 6 of s1, font=\footnotesize] (w1){};
\node[below =-0.3 of w1] (ww'2){};
\node[inner sep=0pt, left = 2 of w1, font=\footnotesize] (w'1){};
\node[below =-0.3 of w'1] (ww'1){};
\node[inner sep=0pt, right = 3cm of w1, font=\footnotesize] (w2){};
\node[below =-0.3 of w2] (ww'N){};
\node[inner sep=0pt, right = 3cm of w2, font=\footnotesize] (w3){};
\node[below =-0.3 of w3] (ww1){};
\node[inner sep=0pt, right = 2cm of w3, font=\footnotesize] (w4){};
\node[below =-0.3 of w4] (ww2){};
\node[inner sep=0pt, right = 3cm of w4, font=\footnotesize] (w5){};
\node[below =-0.3 of w5] (wwN){};
\node[tasks1, below left = 0 and 0 of ww'1] (A1+R1) {\footnotesize$\bfT_{1, 1}^{\text{u}}, \bfB$};
\node[tasks1, right = 0.2 of A1+R1] (A2+R2) {\footnotesize$\bfT_{2,1}^{\text{u}}, \bfB$};
\node[tasks1, right = 0.7 of A2+R2] (AN+RN) {\footnotesize$\bfT_{\noworkersuntrusted, 1}^{\text{u}}, \bfB$};
\node[tasks1, below = 0.1 of A1+R1] (A'1+R'1) {\footnotesize$\bfT_{1,2}^{\text{u}}, \bfB$};
\node[tasks1, below = 0.1 of A2+R2] (A'2+R'2) {\footnotesize$\bfT_{2,2}^{\text{u}}, \bfB$};
\node[tasks1, below = 0.1 of AN+RN] (A'N+R'N) {\footnotesize$\bfT_{\noworkersuntrusted,2}^{\text{u}}, \bfB$};
\node[tasks1, below = 0.4 of A'1+R'1] (Al+Rl) {\footnotesize$\bfT^{\text{u}}_{1, \stragglerun}, \bfB$};
\node[tasks1, below = 0.4 of A'2+R'2] (A'l+R'l) {\footnotesize$\bfT^{\text{u}}_{2, \stragglerun}, \bfB$};
\node[tasks1, below = 0.4 of A'N+R'N] (A''l+R''l) {\footnotesize$\bfT^{\text{u}}_{\noworkersuntrusted, \stragglerun}, \bfB$};
\node[above = 2 of Al+Rl] (W'1) {\footnotesize {\Worker} $w^{\text{u}}_1$};
\node[above = 2 of A'l+R'l] (W'2) {\footnotesize {\Worker} $w^{\text{u}}_2$};
\node[above = 2 of A''l+R''l] (W'N) {\footnotesize {\Worker} $w^{\text{u}}_{\noworkersuntrusted}$};
\node[tasks2, right = 1 of AN+RN] (R1) {\footnotesize$\bfT^{\text{t}}_{1, 1}, \bfB$};
\node[tasks2, right = 0.2 of R1] (R2) {\footnotesize$\bfT^{\text{t}}_{2, 1}, \bfB$};
\node[tasks2, right = 0.7 of R2] (RN) {\footnotesize$\bfT^{\text{t}}_{\noworkerstrusted, 1}, \bfB$};
\node[tasks2, below = 0.1 of R1] (R'1) {\footnotesize$\bfT^{\text{t}}_{1, 2}, \bfB$};
\node[tasks2, below = 0.1 of R2] (R'2) {\footnotesize$\bfT^{\text{t}}_{2, 2}, \bfB$};
\node[tasks2, below = 0.1 of RN] (R'N) {\footnotesize$\bfT^{\text{t}}_{\noworkerstrusted, 2}, \bfB$};
\node[tasks2, below = 0.4 of R'1] (R''1) {\footnotesize$\bfT^{\text{t}}_{1, \stragglertr}, \bfB$};	
\node[tasks2, below = 0.4 of R'2] (R''2) {\footnotesize$\bfT^{\text{t}}_{2, \stragglertr}, \bfB$};
\node[tasks2, below = 0.4 of R'N] (R''N) {\footnotesize$\bfT^{\text{t}}_{\noworkerstrusted, \stragglertr}, \bfB$};
\node[above = 2 of R''1] (W1) {\footnotesize {\Worker} $w^{\text{t}}_1$};
\node[above = 2 of R''2] (W2) {\footnotesize {\Worker} $w^{\text{t}}_2$};
\node[above = 2 of R''N] (WN) {\footnotesize {\Worker} $w^{\text{t}}_{\noworkerstrusted}$};
\node[thick,above right = 0.8 and 0.75 of w1](pikat1){};
\node[thick,above right = 0.8 and 2 of w4](pikat'1){};
\node[thick,below = 1.1 of pikat1](pikat2) {$\dots$};
\node[thick,below = 0.5 of pikat2](pikat3) {$\dots$};
\node[thick,below = 0.7 of pikat3](pikat4) {$\dots$};
\node[thick,below = 1.1 of pikat'1](pikat'2) {$\dots$};
\node[thick,below = 0.5 of pikat'2](pikat'3) {$\dots$};
\node[thick,below = 0.9 of pikat'3](pikat'4) {$\dots$};
\draw[ultra thick] (-2.8,0.5) -- (-2.8, -3.7);
\node[above left = 0.8 and 0 of w1](cluster1) {\footnotesize \emph{Untrusted cluster}};
\node[above right = 0.8 and -0.2 of w4](cluster2) {\footnotesize \emph{Partly trusted cluster}};
\node[left = 0 of A1+R1]{\footnotesize Layer $1$};
\node[left = 0 of A'1+R'1]{\footnotesize Layer $2$};
\node[left = 0 of Al+Rl]{\footnotesize Layer $\stragglerun$};
\node[right = 0 of R''N]{\footnotesize Layer $\stragglertr$};
\node[right = 0 of R'N]{\footnotesize Layer $2$};
\node[right = 0 of RN]{\footnotesize Layer $1$};
     
\end{tikzpicture}}
    \caption{The two non-communicating clusters, namely the untrusted and partly trusted, are illustrated on the left and right-hand side, respectively. The {\workers} of the untrusted cluster $w_1^{\text{u}}, w_2^{\text{u}}, \dots, w_{\noworkersuntrusted}^{\text{u}}$ get each $\stragglerun$ tasks (created as in~\cref{subsec:task_distrib}) that leak nothing about the input matrix. On the other hand, each {\worker} of the partly trusted cluster $w_1^{\text{t}}, w_2^{\text{t}}, \dots, w_{\noworkerstrusted}^{\text{t}}$ gets $\stragglertr$ tasks that leak some information about the input matrix. Every {\worker} has to multiply the designated matrices with the public matrix $\bfB$ and then send each computation back to the {\master}.}
    \label{fig:fractional_repetition_sparse}
\end{figure*}
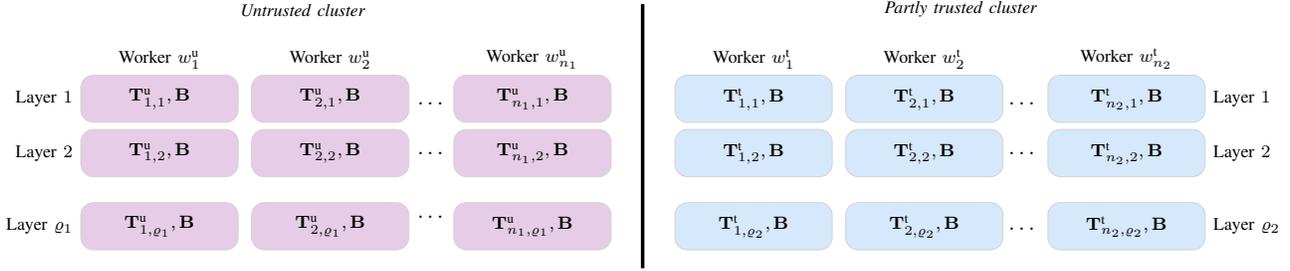

The scheme follows from an immediate application of the scheme constructed in \cref{sec:sss_polynomial}. The input matrices are encoded into shares using our sparse polynomial encoding introduced in \cref{def:model} with the parameters derived in \cref{thm:optimal_pmf_straggler_tol}. Namely, 
the {\master} picks a set of distinct non-zero evaluation points $\{\eval[j]\}_{j=1}^{\nbworkers}$ and creates a polynomial pair $\polyanoind,\polybnoind$ as
\begin{align}
    \label{eq:polys}
    \polyanoind &= \bfA + x \bfR & \text{ and } & &
    \polybnoind &= \bfB + x \bfS,%
\end{align}
where $\bfR$ and $\bfS$ are generated depending on $\bfA$ and $\bfB$, respectively. Each \worker $j$ is then assigned $\polyanoind[\eval[j]]$ and $\polybnoind[\eval[j]]$ and its task is to return the multiplication $\polyanoind[\alpha_j] \cdot \polybnoind[\alpha_j]$ to the {\master}.
Define the degree-$2$ polynomial $\respolynoind[x]$ as
\begin{equation}\label{eq:h_pol}
    \respolynoind = \polyanoind \cdot \polybnoind = \bfA\bfB + x (\bfR \bfB + \bfA\bfS) + x^2 \bfR\bfS\,.
\end{equation}
Then, the {\master} obtains the desired computation $\bfC = \bfA \bfB = \respolynoind[0]$ from the responses of any three \workers. This task assignment strategy allows for \emph{full} straggler tolerance, i.e., $\cC_j = \emptyset$ if {\worker} $j$ is a straggler or $\cC_j = \{h(\eval)\}$, otherwise.

Each \worker gets assigned a computation as large as the multiplication $\bfA \bfB$. To reduce the computation complexity, this strategy can be coupled with cyclic shift task assignment, e.g.,~\cite{coded_sparse_matrix_leverages_partial_stragglers},\cite{el2010fractional}, at the expense of reducing the straggler tolerance. For example, in a system of $\nbworkers$ \workers, the matrices $\bfA$ and $\bfB$ can be divided into $\splitpsmm$ parts, $\splitpsmm\vert \nbworkers$, each. The division of the matrices is column-wise for $\bfA$ and row-wise for $\bfB$ so that the multiplication $\bfC = \bfA\bfB$ is the sum of $\bfA_i\bfB_i$ for $i= 1,\dots,\splitpsmm$. The \workers are grouped into $\splitpsmm$ groups $0,\dots,\splitpsmm-1$ of $\nbworkers/\splitpsmm$ \workers. Each part $\bfA_i, \bfB_i$ is encoded into $\nbworkers/\splitpsmm$ tasks using our scheme and assigned to the workers of the group $i-1$. Each \worker now computes a task that is $\splitpsmm$ times smaller. It is clear that one straggler can be tolerated in the worst case (and $\splitpsmm$ stragglers, one from each group, in the best case). To increase straggler tolerance, the encoding of $\bfA_i, \bfB_i$ are additionally assigned to the workers of the group $i+1\mod \splitpsmm$. The straggler tolerance is increased to $\splitpsmm+1$. The leakage is unaffected as the padding matrices used for each encoding are independent. In fact, the leakage is smaller as each \worker observes a smaller matrix. Additionally, assigning the encoding of $\bfA_i, \bfB_i$ to the workers of the groups $i+2\mod \splitpsmm,\dots, i+x \mod \splitpsmm$ for $2\leq x \leq \splitpsmm-1$, increases the straggler tolerance to $x$ and increases the computation load by a multiplicative factor $x$.

Alternatively, to reduce the computation load for large $\nbworkers$ and tolerate $\sigma$ stragglers, the {\master} divides the input matrices (as above) into $\splitpsmm$ matrices such that $\nbworkers \vert \splitpsmm (\sigma+3)$ and $\sigma<\splitpsmm$. The \master creates $\splitpsmm$ polynomial pairs $(\polya,\polyb)$ to encode the $\bfA_i,\bfB_i$'s and assigns $\sigma+3$ distinct evaluations of each polynomial pair $i \in [\splitpsmm]$ to workers $\sigma+3$ workers indexed by $\mathcal{Z}_i, i \in [\splitpsmm]$, where $\vert \mathcal{Z}_i \vert = \sigma+3$. That is, worker $j$ gets $\splitpsmm (\sigma+3) / \nbworkers$ pairs of sparse sub-matrices $\{(\polya[\eval], \polyb[\eval])\}_{i\in\splitpsmm: j\in \mathcal{Z}_i}$. The sets $\mathcal{Z}_i, i \in [\splitpsmm]$ can, for example, be created by a cyclic repetition scheme similar to \cite{tandon2017gradient}, which ensures that $3$ evaluations of each polynomial $h_i(x)$ are guaranteed to be returned despite having any $\sigma$ full stragglers. Consequently, each \worker computes a $(\sigma+3)/\nbworkers$-fraction of the entire multiplication.

\subsection{Setting of clustered \workers}\label{sec:sp_per_it_priv}

In this setting, we assume that the \workers can be grouped into two non-communicating clusters with the properties explained next. We introduce a one-sided private and sparse matrix multiplication scheme, where the matrix $\bfB$ is publicly available. The main advantage of such schemes is tolerating collusions, i.e., allowing $z>1$.

\subsubsection{The clustering model}\label{subsec:clustering}
The $\nbworkers$ workers, where $\nbworkers = \noworkersuntrusted + \noworkerstrusted$, are divided into two clusters of \workers with the following properties:
\begin{itemize}
    \item \textbf{Untrusted cluster:} The cluster consists of $\noworkersuntrusted$ workers, $w_i^u$, for $i = 1, \dots, \noworkersuntrusted$, which are fully untrusted and can all collude. No information leakage about $\bfA$ should be tolerated, i.e., perfect privacy is required. 
    \item \textbf{Partly trusted cluster:} This cluster is composed of $\noworkerstrusted$ {\workers} $w_i^\text{t}$, for $i=1,\dots, \noworkerstrusted$, that are partly trusted by the {\master}. A small amount of information about $\bfA$ can be leaked to any $z$, $1\leq z < \noworkerstrusted$, colluding {\workers}.
\end{itemize}
The grouping of {\workers} into two \emph{non-communicating} clusters is practically relevant since the untrusted cluster can represent an externally hired computational node pool offered by a single provider. In contrast, the partly trusted cluster might represent a locally available computing pool, that can be partly trusted. We combine the sparse one-time pad of~\cref{lemma:sparse_perfect_ot_pad} with the cyclic shift task assignment~\cite{coded_sparse_matrix_leverages_partial_stragglers},\cite{el2010fractional} to obtain the desired privacy, sparsity, and straggler tolerance.

\subsubsection{Task distribution and straggler mitigation}
\label{subsec:task_distrib}

Given the matrix $\bfA \in \F_q^{\nrowsA \times \ncolsA}$, the {\master} generates a padding matrix $\bfR \in \F_q^{\nrowsA \times \ncolsA}$ following~\cref{constr:sparse_one_time pad} and setting $\pz = \pnz \triangleq p$ and $\pext = \frac{1-p}{q-1}$. The {\master} splits the padded matrix $\bfA + \bfR$ row-wise into $\noworkersuntrusted$ sub-matrices and splits the padding matrix $\bfR$ into $\noworkerstrusted$ sub-matrices as follows 
\begin{align*}
\bfA+\bfR &= \begin{pmatrix}
(\bfA+\bfR)_1\\ 
(\bfA+\bfR)_2\\ 
\vdots\\
(\bfA+\bfR)_{\noworkersuntrusted}
\end{pmatrix} 
&& \text{ and } & \bfR = \begin{pmatrix}
\bfR_1\\ 
\bfR_2\\
\vdots\\
\bfR_{\noworkerstrusted}
\end{pmatrix}\,,
\end{align*}
where $(\bfA+\bfR)_i \in \F_q^{\frac{\nrowsA}{\noworkersuntrusted}\times \ncolsA}$ for $i \in [\noworkersuntrusted]$ and $\bfR_i \in \F_q^{\frac{\nrowsB}{\noworkerstrusted}\times \ncolsB}$ for $i \in [\noworkerstrusted]$. Then, the {\master} sends to every {\worker} $w_i^{\text{u}}$ of the untrusted cluster a sub-matrix $\bfT_{i, 1}^{\text{u}} \triangleq (\bfA + \bfR)_i$ for $i\in [\noworkersuntrusted]$ and to every {\worker} $w_a^{\text{t}}$ of the partly trusted cluster the sub-matrices $\bfT_{a, 1}^{\text{t}} \triangleq \bfR_a$ for $a \in [\noworkerstrusted]$. The collection of the matrices $\bfT_{i, 1}^{\text{u}}$ and $\bfT_{a, 1}^{\text{t}}$ is called the first \emph{layer} of computation tasks. The \workers compute $\bfT_{i, 1} \bfB$ and return the result to the \master. 

Assigning one layer only is not resilient to stragglers. To tolerate $\stragglerun-1$ stragglers in the untrusted cluster and $\stragglertr-1$ in the partly cluster, the \master assigns more computation layers to the workers following a cyclic shifting strategy:
\begin{align*}
\bfT^{\text{u}}_{(i \mod \noworkersuntrusted)+1,j} &= \bfT^{\text{u}}_{i,j-1},\quad i \in [\noworkersuntrusted], j \in [\stragglerun] \setminus \{1\},\\
\bfT^{\text{t}}_{(a \mod \noworkerstrusted)+1, b} &= \bfT^{\text{t}}_{a, b-1},\quad a \in [\noworkerstrusted], b \in [\stragglertr] \setminus \{1\}.
\end{align*}

The system is illustrated in~\cref{fig:fractional_repetition_sparse}. Every {\worker} of the untrusted cluster $w_i^{\text{u}}$, $i\in[\noworkersuntrusted]$, sequentially multiplies and sends $\bfT^{\text{u}}_{i,j}\bfB$ to the \master for $j=1,\dots,\stragglerun$. The same applies to the \workers of the partly trusted cluster. %

The privacy, sparsity, and straggler tolerance guarantees of this task assignment strategy are stated in~\cref{thm:main}. 

\begin{theorem}\label{thm:main}
The cyclic shift assignment strategy explained above provides the following guarantees:

\begin{itemize}
    \item \emph{Privacy guarantees}: In the \emph{untrusted cluster}, perfect privacy of $\bfA$ is guaranteed even if all the {\workers} collude, i.e., $\mutinf(\{\rvT_{i,j}^{\text{u}}\}_{i\in[\noworkersuntrusted],j\in[\stragglerun]};\rvA) = 0$. In the \emph{partly trusted} cluster, if up to $z$ {\workers} collude, then with $\leakage{1}=\mutinf(\rvR_{\{i,j\}};\rvA_{\{i,j\}})$, the leakage about $\bfA$ is given by 
\begin{equation}\label{eq:colluding_trusted}
    \mutinf(\{\rvT^{\text{t}}_{i,j}\}_{i\in\cZ, j\in [\stragglertr]};\rvA) = \min\left\{\dfrac{\stragglertr z}{\noworkerstrusted}, 1\right\}\cdot m \cdot n \cdot \leakage{1}.
\end{equation}
\item \emph{Sparsity:} For a desired maximum leakage $\varepsilon$, i.e., $\mutinf(\{\rvT^{\text{t}}_{i,j}\}_{i\in\cZ, j\in [\stragglertr]};\rvA)\leq \varepsilon$, the maximum sparsity levels allowed for the matrices $\bfA+\bfR$ and $\bfR$ are given by
\begin{align*}
    \sar &= p^\star, \quad \text{ and } \quad\sr = p^\star \frac{sq-1}{q-1}+\frac{1-s}{q-1},
\end{align*}
where 
\begin{equation}
    \label{eq:pstar}
    p^\star = \max_{\mutinf(\{\rvT^{\text{t}}_{i,j}\}_{i\in\cZ, j\in [\stragglertr]};\rvA)\leq \varepsilon} p.
\end{equation}

\item \emph{Straggler tolerance:} The {\master} is able to reconstruct the desired computation after receiving any $K^{\text{u}} \triangleq \frac{-\stragglerun^2 +\stragglerun(2\noworkersuntrusted-1)}{2} +1$ and any $K^{\text{t}} \triangleq \frac{-\stragglertr^2 +\stragglertr(2\noworkerstrusted-1)}{2} +1$ responses from the {\workers} of the untrusted and partly trusted clusters, respectively. %
In other words, partial stragglers and up to $\stragglerun - 1$ and $\stragglertr - 1$ stragglers can be tolerated from the respective clusters.
\end{itemize}

\end{theorem}
\begin{proof}
    The proof by construction is omitted for brevity.
\end{proof}

\begin{figure}[t]
    \centering
    \resizebox{.48\textwidth}{!}{\input{./figures/z_vs_pstar}}
    \caption{The value of $p^\star$ as a function of $z$, the number of colluding workers of the partly trusted cluster, for different privacy constraints reflected by choice of $\relleakage$. We take $\noworkerstrusted = 100$, $\stragglertr = 1$, $\sa =0.93$ and consider the extension finite field $\F_{256}$. The value of $p^\star$ controls the sparsity of the shares, i.e., $\sar = p^\star$ and $\sr \approx 0.93p^*$, as in~\cref{eq:pstar}. The case $\relleakage = 0$ corresponds to perfect privacy. It requires $p^\star = q^{-1}$ for all $ z \in [\noworkerstrusted]$, whereas $\relleakage = 1$ corresponds for the non-private case, hence allowing the choice $p^\star = 1$ for all $ z \in [\noworkerstrusted]$. For $0 < \relleakage < 1$, the increase of $z$ and/or the {decrease of} $\relleakage$ reduces the achievable $\sar$ and $\sr$. 
    }
    \label{fig:my_label}
\end{figure}

We plot in~\cref{fig:my_label} the value of $p^*$, representing the sparsity of the shares, versus $z$, the number of colluding {\workers} of the partly trusted cluster for a varying choice of maximum relative leakage $\relleakage$. As reflected from the theory, decreasing $\relleakage$ decreases $p^*$, meaning that the sparsity levels $\sr, \sar$ should have lower values to guarantee higher privacy constraints. In addition, for a fixed relative leakage $\relleakage$, as $z$ increases, $p^\star$ (and the allowed sparsity $\sr, \sar$) decreases since more {\workers} collude and can have access to a bigger portion of resemblance to the private matrix $\bfA$, which should be made more private.

\section{Matrices with Correlated Entries}\label{sec:discussion}
In this section, we discuss the use of shuffling to break the potential correlation between the entries of sparse matrices and mention further research directions. 

In certain applications, e.g., the matrix entries might be correlated when the sparse matrices represent a face~\cite{wright2008robust}. However, our theory and the proposed schemes assume the entries of $\bfA$ (and $\bfB$) to be independent (see~\cref{def:iid}).
Sparsity-preserving pre-processing techniques can be applied to break the correlation between the entries of sparse matrices. In this section, we discuss the effect of applying random permutations to entries of $\bfA$ and $\bfB$. %
We will focus on cases where $\bfA$ and $\bfB$ must remain private. The ideas generalize to cases in which $\bfB$ is public.

Shuffling a matrix is equivalent to multiplying it with permutation matrices from the left and the right.
Those permutation matrices are private to the \master and are unknown to the workers. Considering input matrices $\bfA \in \Fq^{{\nrowsA} \times \ncolsA}$ and $\bfB \in \Fq^{\nrowsB \times {\ncolsB}}$, the {\master} picks three permutation matrices $\bfP_{1} \in \F_2^{{\nrowsA} \times {\nrowsA}}$, $\bfP_{2}\in \F_2^{\ncolsA \times \ncolsA}$ and $\bfP_{3}\in \F_2^{{\ncolsB} \times {\ncolsB}}$. Each column of $\bfP_{i}$ is generated by sampling a random unit vector without replacement from the respective field. The {\master} then conducts the following pre-processing operations
\begin{align*}
    \bfA^\prime &= \bfP_{1} \cdot \bfA \cdot \bfP_{2} \, \text{ and } \, \bfB^\prime = \bfP_{2}^{\mathsf{T}} \cdot \bfB \cdot \bfP_{3}.
\end{align*}
Note that permutation matrices are orthogonal, i.e., the inverse equals their transpose. The matrices $\bfA^\prime,\bfB^\prime$ can then be encoded as in~\cref{subsec:sp_pr_polynomial}. 

\begin{figure}[t]
    \begin{minipage}{\linewidth}
        \centering
        \begin{minipage}[t]{0.45\textwidth}
            \subfloat[Matrix $\bfA$, $\sa \approx0.94$]{\includegraphics[width=\linewidth]    {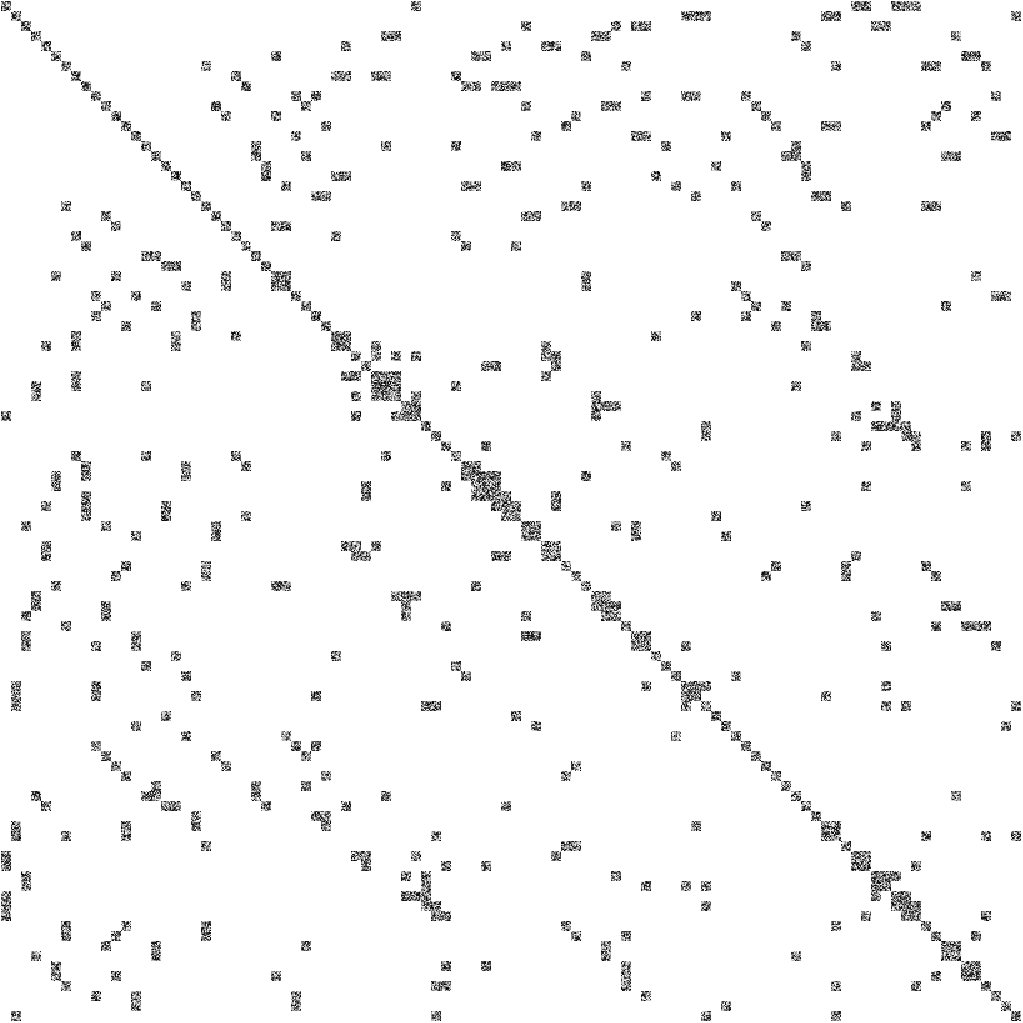}\label{subfig:orig}}
        \end{minipage}%
     \hfill
        \begin{minipage}[t]{0.45\textwidth}
            \subfloat[Share $f(\eval)$ of $\bfA$ with sparsity $s_d \approx 0.85$]{\includegraphics[width=\linewidth]{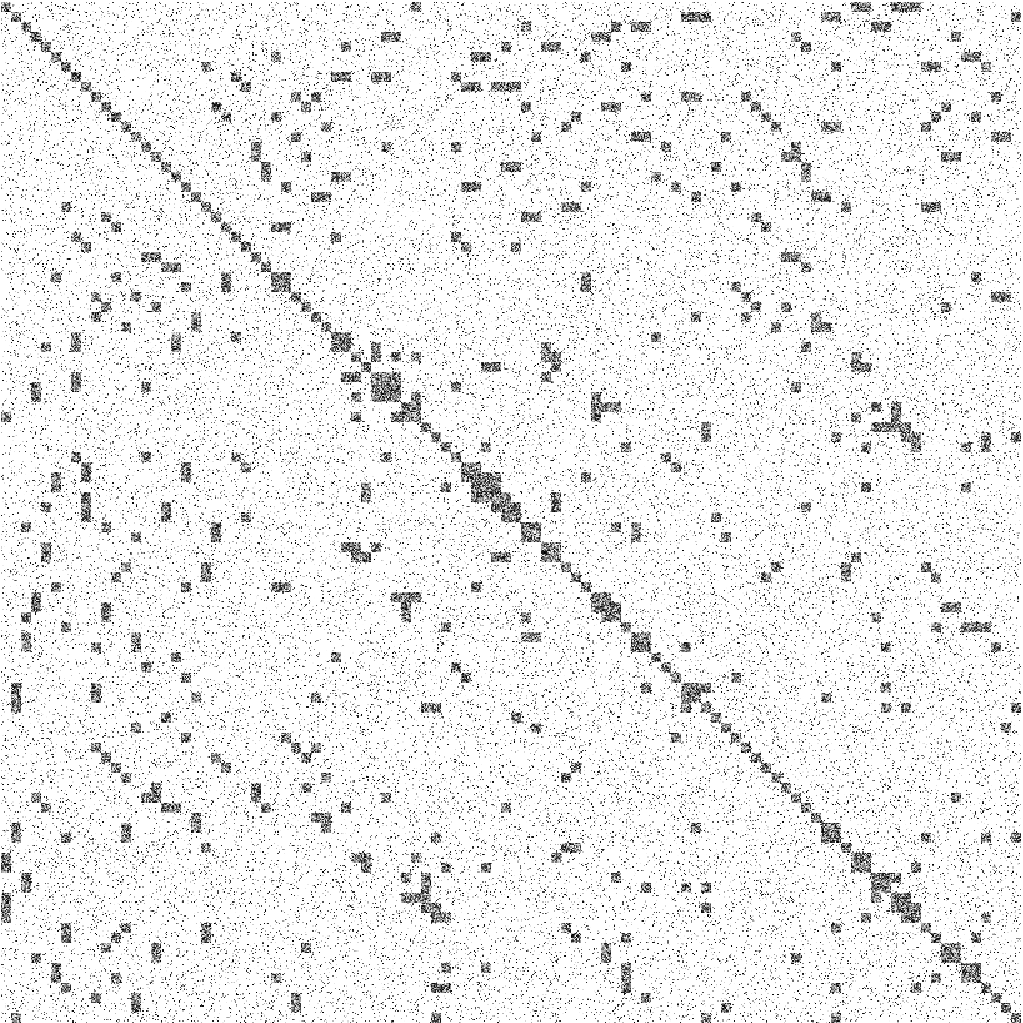}\label{subfig:share}}
        \end{minipage}
    \end{minipage}

    \begin{minipage}{\linewidth}
    \centering
        \begin{minipage}{0.45\textwidth}
            \subfloat[Matrix $\bfA^\prime$ after pre-processing, $s_{\mathbf{A}^\prime}=\sa \approx0.94$]{\includegraphics[width=\linewidth]{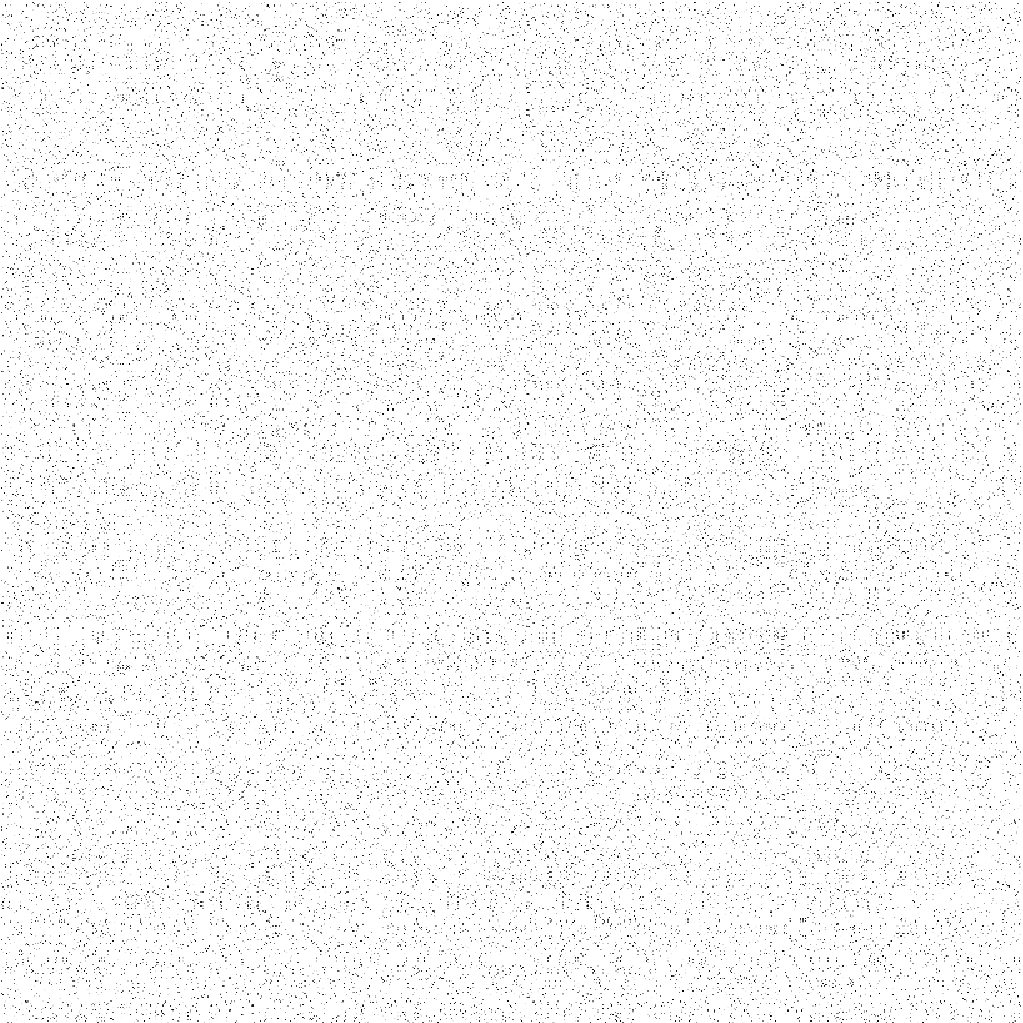} \label{subfig:shrouded}}
        \end{minipage}%
        \hfill%
        \begin{minipage}{0.45\textwidth}
            \subfloat[Share $f(\eval)$ of $\bfA^\prime$ with sparsity $s_d \approx 0.85$]{\includegraphics[width=\linewidth]{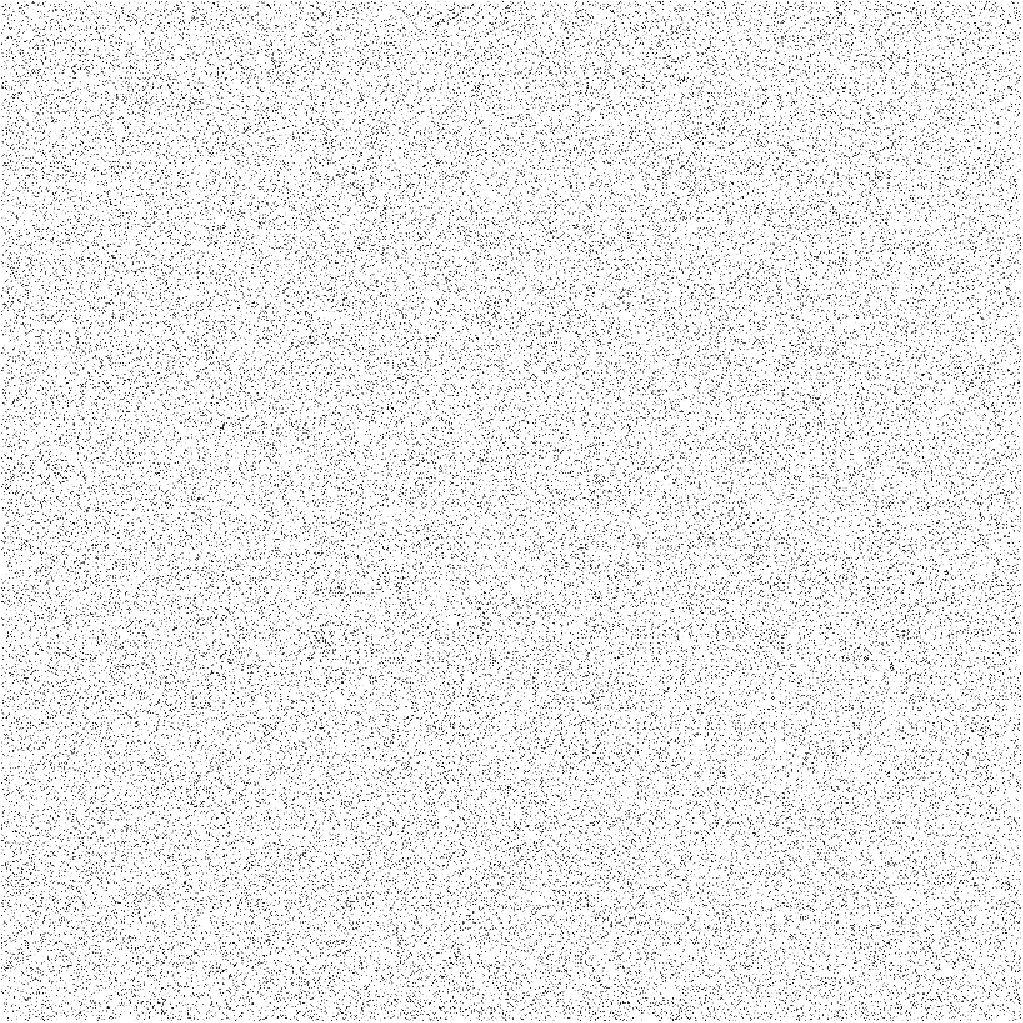}\label{subfig:share_perm}}
            \end{minipage}
    \end{minipage}
    \caption{A depiction of the impact of random permutations on the correlated entries of $\bfA$. The matrix with correlated entries is shown in \cref{subfig:orig}. Naively applying our sparse secret sharing scheme randomizes the values of the entries of $\bfA$ but does not break the correlation between the entries of the shares, cf. \cref{subfig:share}. The effect of applying the random permutation is shown in \cref{subfig:shrouded,subfig:share_perm} where the correlation between the entries is graphically broken.}
    \label{fig:random_permutations}
\end{figure}

The random permutation breaks most correlations between the entries of the input matrices. The colluding \workers will need to revert the permutation to leverage the statistical correlation between the entries of the private matrices. However, information-theoretic privacy guarantees cannot be given since the entries of $\bfA^\prime$ and $\bfB^\prime$ might still have some correlation. %

In~\cref{fig:random_permutations}, we depict a random permutation of an input matrix $\bfA_i \in \mathbb{F}_{89}^{816 \times 816}$ with sparsity $\sa\approx 0.94$ that follows a distribution similar to the one in~\cref{eq:PMF_A}, but with correlated entries. After pre-processing is applied, the correlations are graphically broken (see~\cref{subfig:orig} and~\cref{subfig:shrouded}). The process from \cref{subfig:orig} to \cref{subfig:share} represents the encoding phase as in~\cref{eq:polys} without shuffling. The correlation between the entries of the shares is shown graphically and can be leveraged by the colluding \workers after the input matrix is shuffled and encoded (\cref{subfig:shrouded} and \cref{subfig:share_perm}), then graphically, the correlation between the entries is broken.

This discussion shows that additional care must be taken when encoding matrices with correlated entries. Future research can be conducted to construct sparse secret-sharing schemes that break the correlation between the entries. Alternatively, privacy guarantees of sparsity-preserving pre-processing techniques can be analyzed.

\section{Conclusion}\label{sec:conclusion}
In this work, we opened the door to the coexistence of sparsity and privacy in secret sharing and distributed computing. We focused on the no-collusion regime, i.e., $z=1$, and assumed that the entries of the private matrices are independent and identically distributed. We showed a fundamental tradeoff between privacy and sparsity in the shares. Hence, insisting on perfect information-theoretic privacy does not allow for sparsity in the shares. We constructed sparse and private secret sharing schemes that take as input the desired sparsity and output shares with the provably lowest leakage possible. Using that, we constructed sparse and private matrix multiplication schemes for different settings of interest, some of which also include collusions. We discussed the correlation effect between the entries on the introduced methods.

Generalizing our results for $z>1$ remains an interesting open problem. Similarly, it is interesting to analyze the effect of lifting the assumption of independent entries and to construct schemes tailored for this setting.

\bibliographystyle{ieeetr}
\bibliography{IEEEabrv,Sparsejournal}

\appendix
\subsection{Proof of~\cref{lemma:perfectprivacy}}
\label{proof1}

\newcommand{\rvshare}[1][i]{\ensuremath{\mathrm{S}_{#1}}}
\begin{proof}
We focus only on the element-wise leakage since $\rvA$ has i.i.d. entries. The following analysis holds for any entry $\Aij$ and $\Rij$ of $\bfA$ and $\bfR$. Hence, we abuse notation and drop the $i,j$ index, i.e., we use $\rvA$ for $\Aij$ and $\rvR$ for $\Rij$. Let $\rvshare \triangleq \rvA+\coeff[i]\rvR$ be the corresponding entry of share $i \in [\nbworkers]$. It holds that %
\begin{equation}\label{eq:PMFshare}
\Pr_{\rvshare}(s) = \sum_{a=0}^{q-1} \Pr_\rvA(a)\Pr_{\rvR\vert\rvA} (\coeff[i]^{-1}(s-a)\vert a), \forall s \in \F_q \, .
\end{equation}
Then the entropy $\entropy(\rvshare)$  can be written as \footnote{We consider that $0\log(0) = 0$ for ease of notation.}
\begin{align*}
\entropy(\rvshare) &\overset{(a)}{=} - \sum_{s_j\in\F_q} \Pr_{\rvshare} (s_j) \log_q\left(\Pr_{\rvshare} ( s_j)\right) \\
&\overset{(b)}{=}  - \sum_{s_j\in\F_q} \left(\sum_{a\in\F_q} \Pr_{\rvA}(a)\Pr_{\rvR\vert \rvA} (\coeff[i]^{-1}(s_j-a)\vert a)\right)
\\&\qquad\log_q\left( \dfrac{\sum_{a\in\F_q} \Pr_{\rvA}(a)\Pr_{\rvR\vert\rvA}(\coeff[i]^{-1}(s_j-a)\vert a)}{\sum_{a\in\F_q} \Pr_{\rvA}(a)} \right) \\
&\overset{(c)}{\geq}  - \sum_{s_j\in\F_q} \sum_{a\in\F_q} \Pr_{\rvA}(a)\Pr_{\rvR\vert \rvA} (\coeff[i]^{-1}(s_j-a)\vert a) \\
&\qquad\log_q\left( \Pr_{\rvR\vert\rvA}( \coeff[i]^{-1}(s_j-a)\vert a) \right) \\
&\overset{(d)}{=}  \sum_{a\in\F_q} \Pr_{\rvA}(a) \entropy(\rvR \vert \rvA=a) \\
&= \entropy(\rvR \vert \rvA),
\end{align*}
where $(a)$ is the definition of $q-$ary entropy, $(b)$ is obtained by using \cref{eq:PMFshare} and dividing the argument of the logarithm function by $\sum_{t\in\F_q} \Pr_{\rvA}(t) = 1$, $(c)$ follows from the log-sum inequality and $(d)$ is obtained by grouping the summation terms in $(c)$. 

The log-sum inequality implies that $(c)$ is satisfied with equality if and only if $\Pr_{\rvR\vert \rvA}(\coeff[i]^{-1}(s_j-a)\vert a)$ is constant for every $a \in \F_q$, which has to hold for all $s_j \in \F_q$. That is, there exist $q$ disjoint sets (one for each $s_j\in \F_q$) $\big\{\Pr_{\rvR\vert \rvA}(\coeff[i]^{-1}(s_j-a)\vert a)\big\}_{a\in \F_q}$ whose entries must take the same value. %

Given two shares $\rvshare[1]$ and $\rvshare[2]$ with distinct evaluation points $\coeff[1]$ and $\coeff[2]$, we have two of such sets. For a fixed $s_1 \in \F_q$ (in the equation for $\rvshare[1]$) and for every $s_2\in \F_q$, there exists a tuple $(s_1, s_2, \bar{a})$ with $\bar{a}\in \F_q$ such that $\coeff[1]^{-1}(s_1-\bar{a}) = \coeff[2]^{-1}(s_2-\bar{a})$. In particular, the value of $\bar{a}$ can be calculated as

\begin{align*}
    \coeff[1]^{-1}(s_1-\bar{a}) &= \coeff[2]^{-1}(s_2-\bar{a}) \\
    \underbrace{\coeff[1]^{-1} \coeff[2]}_{\triangleq c \neq 0} s_1 + \bar{a} (1-\coeff[1]^{-1} \coeff[2]) &= s_2 \\[-.6cm]
    \bar{a} &= (s_2 - c s_1)/(1-c).
\end{align*}
This particular choice $\bar{a}$ shows that $$\big\{\Pr_{\rvR\vert \rvA}(\coeff[1]^{-1}(s_1-a)\vert a)\big\}_{a\in \F_q} \cap \big\{\Pr_{\rvR\vert \rvA}(\coeff[2]^{-1}(s_2-a)\vert a)\big\}_{a\in \F_q} \neq \emptyset.$$ 
Hence the elements of both sets must be equal, i.e., 
\begin{equation*}
\Pr_{\rvR\vert \rvA}(\coeff[1]^{-1}(s_1-a)\vert a) = \Pr_{\rvR\vert \rvA}(\coeff[2]^{-1}(s_2-a)\vert a)\quad \forall a\in \F_q,
\end{equation*}
to fulfill (c) with equality. Moving forward, by fixing $s_1$ one can find such intersections for all $s_2 \in \F_q$. Since the sum over $s_2$ in (c) covers constraints on the whole space of possible probability values $\Pr_{\rvR\vert \rvA}(r \vert a)$, we can link the constraints of the $q$ disjoint sets $\big\{\Pr_{\rvR\vert \rvA}(\coeff[2]^{-1}(s_2-a)\vert a)\big\}_{a\in \F_q}$ for all $s_2$ through the set $\big\{\Pr_{\rvR\vert\rvA}(\coeff[1]^{-1}(s_1-a)\vert a)\big\}_{a\in \F_q}$ for any fixed $s_1$. Consequently, all entries in the set $\big\{\Pr_{\rvR\vert\rvA}(\coeff[1]^{-1}(s-a)\vert a)\big\}_{a, s\in \F_q} = \big\{\Pr_{\rvR\lvert\rvA}(r\vert a)\big\}_{a, r\in \F_q}$ have to be equal to fulfill (c) with equality. The same steps can be repeated for more than two shares.
We can further write the leakage as
\begin{align*}
\mutinf(\rvA; \rvshare[i]) &= \entropy(\rvA+\coeff[i]\rvR) - \entropy(\rvA+\coeff[i]\rvR \lvert \rvA) \\  &= \entropy(\rvshare[i]) - \entropy(\rvR \vert \rvA) \geq 0,
\end{align*}
where equality is only given if (c) is fulfilled with equality, and hence $\rvR$ is independent of $\rvA$, i.e., $\Pr_{\rvR\vert\rvA}(r\vert a) = \Pr_{\rvR}(r)$, and is uniform.
\end{proof}

\subsection{Proof of~\cref{thm:optimal_pmf}}
\label{proof6}
\begin{proof}
We will use the method of Lagrange multipliers to combine the total leakage as quantified in~\cref{lemma:optimization} with the constraints in \cref{eq:constraint1,eq:constraint2,eq:constraint3,eq:constraint4}. The combination, called as the objective function that has to be minimized is quantified as 
\begin{align*}
    \lossabbrev &\triangleq \loss \\
    & = \totalleakage + \lambda_1 c_1(\pz,\pzinv) + \\
    & + \lambda_2 c_2(\pext,\pnz,\pnzinv) + \lambda_3 c_3(\pz,\pext) + \lambda_4 c_4(\pz,\pnz).
\end{align*}

This objective can be minimized by setting the gradient to zero, i.e., $\grad \lossabbrev = 0$, which is equivalent to solve a system of nine equations with nine unknowns. Simplifying $\nabla_{\pz,\pzinv,\pext,\pnz,\pnzinv} \lossabbrev = 0$ yields
\begin{align}
    \pz (\pnzinv)^2 &= \pzinv \pext \pnz. \label{eq:grad_obj_compact}
\end{align}

Using the constraints~\cref{eq:constraint1,eq:constraint2,eq:constraint3,eq:constraint4}, one can express $\pzinv, \pext, \pnz,\pnzinv$ as polynomials of degree one in $\pz$, and plugging them into \cref{eq:grad_obj_compact} transfers the problem into finding the root of a polynomial of degree $3$ in $\pz$ such that it satisfies the constraints $\sar-1+s \leq \pz s \leq \sar$ and $\sr-1+s \leq \pz s \leq \sr$. Let $\qfac \triangleq \frac{(q-2)^2}{q-1}$, then the polynomial denotes as $\alpha \pz^3 + \beta \pz^2 + \gamma \pz + \delta = 0$, where $\alpha = s^2 (4+\qfac)$, $\beta = 4s(1-s-\sar-\sr)-\qfac s (\sr+\sar+s)$, $\gamma = (1-s-\sar-\sr)^2+\qfac(s\sar + s\sr + \sar \sr)$ and $\delta = -\qfac \sar \sr$. Having $\pz$, the remaining unknowns $\pzinv,\pext,\pnz$ and $\pnzinv$ follow directly from \eqref{eq:constraint1}-\eqref{eq:constraint4}. Note that the problem is convex, then the solution is a global minima.
\end{proof}

\subsection{Sparsity of Two Shares Minimizing the Leakage}
\label{proof5}
\begin{lemma}
    \label{lemma:optimaldelta}
    Given a desired average sparsity level $\savg$, then the minimum total leakage is obtained for $\sr = \sar = \savg$.%
\end{lemma}

\begin{proof}
    Given a desired average sparsity level $\savg = \frac{\sa+\sar}{2}$, we define $\sdiff$ such that $\sr = \savg -\sdiff$ and $\sar = \savg + \sdiff$ and determine the optimal of $\sdiff$ that minimizes the leakage. We utilize the method of Lagrange multipliers to combine the total leakage from \cref{lemma:optimization}, with the constraints in \cref{eq:constraint1,eq:constraint2,eq:constraint3,eq:constraint4}. Note that~\cref{eq:constraint3,eq:constraint4} have to be slightly updated by plugging in $\sr = \savg -\sdiff$ and $\sar = \savg + \sdiff$. Thus, the objective function to be minimized is a function of $\sdiff$ and can be expressed as
\begin{align*}
    \lossabbrev &\triangleq \totalleakagesdiff + \lambda_1 c_1(\pz,\pzinv) + \\
    & + \lambda_2 c_2(\pext,\pnz,\pnzinv) + \lambda_3 c_3(\pz,\pext) + \lambda_4 c_4(\pz,\pnz).
\end{align*}

This objective is minimized by setting the gradient to zero, i.e., $\grad \lossabbrev = 0$, which is equivalent to solving a system of eleven equations with eleven unknowns. %

This is a convex optimization problem, since the objective function is a sum of convex functions and the constraints are affine in the unknowns. The function $\totalleakagesdiff$ is convex in $\sdiff$ as shown in the composition theorem~\cite[Ch. 3.2.4]{boyd2004convex} since $-\log(x)$ is convex and $g(x) = \savg \pm \sdiff$ is convex. Thus, $-\log(\savg \pm \sdiff)$ is convex. %

We now state the system of equations given by $\nabla_{\pext,\pnz,\sdiff} \lossabbrev = 0$:
\begin{align}
    \frac{\partial \lossabbrev}{\partial \sdiff} &= s \bigg[ \pz \left(-\frac{1}{\sar} + \frac{1}{\sr}\right) + \nonumber \\
    & + (q-1)\pzinv \left( \frac{1}{1\!-\!\sar} - \frac{1}{1\!-\!\sr}\right)\bigg] + \nonumber \\
    & + (1-s) \! \bigg[ p_2 \! \left( \frac{1}{1\!-\!\sar} \! + \! \frac{1}{\sr} \right) \! - \! p_3 \left( \frac{1}{\sar} \! + \! \frac{1}{1\!-\!\sr} \right) \nonumber \\
    & + (q-2) \ptwoinv \! \left( \frac{1}{1\!-\!\sar} - \frac{1}{1\!-\!\sr}\right) \bigg] \! + \! \lambda_3 \! - \! \lambda_4 = 0 \label{eqline:sys1} \\
    \frac{\partial \lossabbrev}{\partial p_2} &= (1-s)\left(2\log(p_2) - \log(\sr) - \log(\sarinv) + 2\right)  \nonumber \\
    &+ \lambda_2 + (1-s) \lambda_3 = 0 \label{eqline:partial_p2} \\
    \frac{\partial \lossabbrev}{\partial p_3} &= (1-s)\left(2\log(p_3) - \log(\srinv) - \log(\sar) + 2\right) \nonumber \\
    &+ \lambda_2 + (1-s) \lambda_4 = 0 \label{eqline:partial_p3}
\end{align}
We will first calculate $\lambda_3 - \lambda_4$ from \eqref{eqline:partial_p2} and \eqref{eqline:partial_p3} to be later substituted in \eqref{eqline:sys1}. By computing $\eqref{eqline:partial_p2} - \eqref{eqline:partial_p3}$ and dividing the result by $(1-s)$, we obtain
\begin{align}
    &\lambda_3-\lambda_4 = \nonumber \\ 
    &= 2 (\log(p_3) - \log(p_2)) + \log\left(\frac{\sarinv}{\sar}\right) + \log\left(\frac{\sr}{\srinv}\right) \! \label{eqline:lambda_diff1} \\
    &= 2\log\left(\frac{\savg+\sdiff-s p_1}{1-s}\right) - 2\log\left(\frac{\savg-\sdiff-s p_1}{1-s}\right) + \nonumber \\
    &+ \log\left(\frac{1-\savg-\sdiff}{(q-1)(\savg+\sdiff)}\right) + \log\left(\frac{(q-1)(\savg-\sdiff)}{1-\savg+\sdiff}\right) \label{eqline:lambda_diff2} \\ 
    &= 2 \log\! \left(\! \frac{\savg+\sdiff-s p_1}{\savg-\sdiff-s p_1} \! \right) \!\! + \! \log\!\left(\! \frac{(1-\savg-\sdiff)(\savg-\sdiff)}{(1-\savg+\sdiff)(\savg+\sdiff)} \! \right) %
    \label{eqline:lambda34}
\end{align}
where from \eqref{eqline:lambda_diff1} to \eqref{eqline:lambda_diff2} we used \eqref{eq:constraint3} and \eqref{eq:constraint4} to express $p_2$ and $p_3$ in terms of $p_1$, respectively.
Simplifying \eqref{eqline:sys1} by use of $\nabla_{\lambda_1\lambda_2,\lambda_3,\lambda_4,\lambda_5} \lossabbrev = 0$, i.e., the constraints in \eqref{eq:constraint1}-\eqref{eq:constraint4} yields
\begin{align}
    &\frac{\partial \lossabbrev}{\partial \sdiff} - (\lambda_3 - \lambda_4) \nonumber \\
    & = \frac{1}{\sr} (s p_1 + (1-s) p_2) - \frac{1}{\sar} (s p_1 + (1-s)p_3)  + \nonumber \\
    & + \frac{1}{1-\sar} \underbrace{(s(q-1) \pzinv + (1-s) p_2 \! + \! (1-s)(q-2)\ptwoinv)}_{(a)} \nonumber \\
    & - \frac{1}{1-\sr} \underbrace{((q-1)\pzinv + (1-s)p_3 + (q-2)(1-s)\ptwoinv)}_{(b)} \nonumber \\
    &= \frac{\sr}{\sr} -\frac{\sar}{\sar} + \frac{1-\sar}{1-\sar} - \frac{1-\sr}{1-\sr} = 0, \label{eqline:partial_sdiff_simplified}
\end{align}
where we used that $\sr = s p_1 + (1-s) p_2$ from \eqref{eq:constraint3}, $\sar = s p_1 + (1-s) p_3$ from \eqref{eq:constraint4}. Further, we obtain $1-\sar = (a)$ by solving \eqref{eq:constraint1} for $p_1$, substituting $p_1$ in \eqref{eq:constraint4} and solving for $p_3$ and finally inserting in \eqref{eq:constraint2}, multiplying by $(1-s)$ and solving for $1-\sar$. Lastly, we obtain $1-\sr = (b)$ by solving \eqref{eq:constraint1} for $p_1$, substituting $p_1$ in \eqref{eq:constraint3} and solving for $p_2$ and finally inserting in \eqref{eq:constraint2}, multiplying by $(1-s)$ and solving for $1-\sr$.
Putting \eqref{eqline:lambda34} and \eqref{eqline:partial_sdiff_simplified} together, we obtain
\begin{align*}
    \frac{\partial \lossabbrev}{\partial \sdiff} + \lambda_3 - \lambda_4 = 2& \log\left(\frac{\savg+\sdiff-s p_1}{\savg-\sdiff-s p_1} \right) + \\
    &+ \log\left(\frac{(1-\savg-\sdiff)(\savg-\sdiff)}{(1-\savg+\sdiff)(\savg+\sdiff)} \right) = 0. 
\end{align*}
Sine the optimization is convex, it follows that $\sdiff=0$ minimizes the objective in \cref{lemma:optimization}. %
\end{proof}

\subsection{Proof of~\cref{thm:optimal_pmf_straggler_tol}}
\label{proof8}
\begin{proof}
Being a convex optimization problem with affine constraints, we utilize the method of Lagrange multipliers \cite{rockafellar1993lagrange} to combine the leakage in \cref{eq:leakage_straggler_tol} with the constraints in \cref{eq:gconstraint1strag,eq:gconstraint2strag,eq:gconstraint3strag}. Then the objective function is
\begin{align}
    \lossabbrev \triangleq & \totalleakagestragglers  + \lambda_1 c_1(\pz,\pzinv) \nonumber \\ %
    &+ \lambda_2 c_2(\pc,\pcinv) + \lambda_3 c_3(\pz,\pc), \label{eq:leakage}
\end{align}
where $\totalleakagestragglers$ is the leakage defined in the optimization problem in \cref{lemma:leakage_straggler_tol}, which results from~\cref{def:model} (similar to the~\cref{lemma:optimization}). Then the minimization problem of \cref{lemma:leakage_straggler_tol} can be solved by setting the gradient of the objective in \eqref{eq:leakage} to zero, i.e., $\gradstragglers \lossabbrev = 0$, which amounts to solving a system of seven equations with seven unknowns. We utilize a subset of those equations
and one can show that given the objective $\lossabbrev$, we obtain from $\nabla_{\pz,\pzinv,\pc,\pcinv} \lossabbrev = 0$ the relation
\begin{align}
    \pz (\pcinv)^\shares &= \pzinv \pc^\shares. \label{eq:grad_obj_compact}
\end{align}
To obtain this result, we use that $\frac{\partial x\log(\frac{x}{y})}{\partial x} = \log(x)-\log(y)+1$. Then, from \eqref{eq:leakage} and the objective and constraints in~\cref{lemma:leakage_straggler_tol}%
, we obtain the following non-linear system of equations for $\nabla_{\pz,\pzinv,\pc,\pcinv} \lossabbrev = 0$:
\begin{align*}
    \frac{\partial \lossabbrev}{\partial p_1} &= s\left(\log(p_1) - \log(\sd) + 1 \right) + \lambda_1 + s\lambda_3 = 0, \\
    \frac{\partial \lossabbrev}{\partial p_1^{\text{inv}}} &= s(q-1)\left(\log(p_1^{\text{inv}}) - \log(\sdinv) + 1\right) + (q-1) \lambda_1 = 0, \\
    \frac{\partial \lossabbrev}{\partial \pc} &= \begin{aligned}[t] & (1-s)\Big(\log(\pc) - \log(\sd) + 1 +  (\shares-1)\\
    & \cdot  ( \log (\pc) - \log(\sdinv) + 1 ) \Big) \! + \! \shares \lambda_2 + (1-s) \lambda_3 = 0, \end{aligned} \\
    \frac{\partial \lossabbrev}{\partial \pcinv} &= \begin{aligned}[t] & (1-s)(q-\shares) \left(\log(\pcinv) - \log(\sdinv) + 1\right) \\
    & + (q-\shares) \lambda_2 = 0. \end{aligned}
\end{align*}
By scaling, we simplify and reformulate this system of equations and obtain
\begin{align}
    0 &= \log(p_1) - \log(\sd) + 1 + \lambda_1^\prime + \lambda_3^\prime, \label{eqline:stragpartial_p1} \\
    0 &= -\log(p_1^{\text{inv}}) + \log(\sdinv) - 1 - \lambda_1^\prime, \label{eqline:stragpartial_p1inv} \\
    0 &= \begin{aligned}[t] - \shares \log(\pc) + &\log(\sd) + (\shares-1) \log(\sdinv) \\& - \shares - \shares \lambda_2^\prime - \lambda_3, \end{aligned} \label{eqline:stragpartial_pc} \\
    0 &= \shares \log(\pcinv) - \shares \log(\sdinv) + \shares + \shares \lambda_2^\prime, \label{eqline:stragpartial_pcinv}
\end{align}
where $\lambda_1^\prime = \lambda_1/s$, $\lambda_2^\prime = \lambda_2/(1-s)$ and $\lambda_3^\prime = \lambda_3/s$. 
Summing \eqref{eqline:stragpartial_p1}-\eqref{eqline:stragpartial_pcinv} 
leads to $\log\left(\pz {\pcinv}^\shares\right) - \log\left(\pzinv \pc^\shares\right)=0$. Exponentiating both sides yields the result stated in \eqref{eq:grad_obj_compact}. The condition on the root $\pc$ is a consequence of \eqref{eq:gconstraint2strag} and \eqref{eq:gconstraint3strag}.

Then we use \eqref{eq:grad_obj_compact} and $\nabla_{\lambda_1, \lambda_2, \lambda_3} \lossabbrev = 0$, where the latter system of equations is linear and equivalent to \eqref{eq:gconstraint1strag}-\eqref{eq:gconstraint3strag}. We solve \cref{eq:gconstraint1strag,eq:gconstraint2strag,eq:gconstraint3strag} for $\pz,\pzinv,\pcinv$ as a function of $\pc$ to obtain
\begin{align}
    (q-1) \frac{\sd-(1-s)\pc}{s-\sd+(1-s)\pc} = \left((q-\shares)\frac{\pc}{1-\shares\pc}\right)^\shares\!\!\!\!. \label{eq:grad_obj_compact_stragglers}
\end{align}
We reformulate \eqref{eq:grad_obj_compact_stragglers} to get the following polynomial in $\pc$ with degree $\shares+1$:
\begin{align*}
    -\pc^{\shares+1} - \! s_2 \pc^\shares + c s_1 \! \sum_{k=1}^\shares \binom{\shares}{k} (-\shares \pc)^k - \pc c \! \sum_{k=0}^\shares (-\shares \pc)^k.
\end{align*}
Sorting the coefficients by powers of $\pc$ results in the root finding problem given in \cref{thm:optimal_pmf_straggler_tol}. %
The results for $\pz,\pzinv$ and $\pcinv$ can be obtained from $\nabla_{\lambda_1,\lambda_2,\lambda_3} \lossabbrev = 0$, i.e., from the constraints stated in \cref{eq:gconstraint1strag,eq:gconstraint2strag,eq:gconstraint3strag}. The condition on the root $\pc$ is a direct consequence of \cref{eq:gconstraint2strag,eq:gconstraint3strag}. 
\end{proof}

\balance

\end{document}